\documentclass[journal,twocolumn,10pt]{IEEEtran}
\normalsize
\renewcommand{\baselinestretch}{.9875}
\IEEEoverridecommandlockouts
\makeatletter
\def\ps@IEEEtitlepagestyle{
  \def\@oddfoot{\mycopyrightnotice}
  \def\@evenfoot{}
}
\def\mycopyrightnotice{
  {\footnotesize
  \begin{minipage}{\textwidth}
  \centering
  Copyright~\copyright~2023 IEEE. Personal use of this material is permitted. However, permission to use this  \\ 
  material for any other purposes must be obtained from the IEEE by sending a request to pubs-permissions@ieee.org.
  \end{minipage}
  }
}
\usepackage{cite}
\usepackage{amsmath,amssymb,amsfonts,mathrsfs,amsthm}
\usepackage{graphicx}
\usepackage{textcomp}
\usepackage{xcolor}
\usepackage{url}
\usepackage{comment}
\usepackage{multirow}
\usepackage[inkscapeformat=eps]{svg}
\usepackage{algorithm}
\usepackage{algpseudocode}
\usepackage{hyperref}
\hypersetup{
    colorlinks=true,
    linkcolor=teal,
    filecolor=magenta,      
    urlcolor=cyan,
    citecolor=teal
    }
\usepackage{adjustbox}
\usepackage[font=small]{caption}
\usepackage{subcaption}
\usepackage{orcidlink}
\newtheorem{theorem}{Theorem}
\newtheorem{remark}{Remark}
\def\BibTeX{{\rm B\kern-.05em{\sc i\kern-.025em b}\kern-.08em T\kern-.1667em\lower.7ex\hbox{E}\kern-.125emX}}

\DeclareMathOperator*{\argmin}{arg\,min}
\begin{document}

\title{Knowledge-Aware Semantic Communication System Design and Data Allocation}

\author{Sachin~Kadam$^{\orcidlink{0000-0001-7085-3365}}$~and~Dong~In~Kim$^{\orcidlink{0000-0001-7711-8072}}$,~\IEEEmembership{Fellow,~IEEE}
\thanks{A preliminary version of this paper is published in IEEE International Conference on Communications (ICC) 2023~\cite{kadam2023knowledge}.}
\thanks{S.~Kadam and D.~I.~Kim are with the Department of Electrical and Computer Engineering, Sungkyunkwan University (SKKU), Suwon 16419, Republic of Korea (e-mail: sachinkadam@skku.edu, dikim@skku.ac.kr).}}
\maketitle

\begin{abstract}
The recent emergence of 6G raises the challenge of increasing the transmission data rate even further in order to overcome the Shannon limit. Traditional communication methods fall short of the 6G goals, paving the way for Semantic Communication (SemCom) systems that have applications in the metaverse, healthcare, economics, etc. In SemCom systems, only the relevant keywords from the data are extracted and used for transmission. In this paper, we design an auto-encoder and auto-decoder that only transmit these keywords and, respectively, recover the data using the received keywords and the shared knowledge. This SemCom system is used in a setup in which the receiver allocates various categories of the same dataset collected from the transmitter, which differ in size and accuracy, to a number of users. This scenario is formulated using an optimization problem called the data allocation problem (DAP). We show that it is NP-complete and propose a greedy algorithm to solve it.  Using simulations, we show that the proposed methods for SemCom system design outperform state-of-the-art methods in terms of average number of words per sentence for a given accuracy, and that the proposed greedy algorithm solution of the DAP performs significantly close to the optimal solution.
\end{abstract}


\section{Introduction}\label{Sec:Intro}
As per the prediction in~\cite{rajatheva2020white}, semantic communication (SemCom) technology is identified as one of the key ingredients in 6G due to the requirement of low latency and high data rate transmissions. The recent emergence of SemCom technologies finds applications in wide range of fields such as economics~\cite{liew2022economics}, metaverse~\cite{ismail2022semantic}, autonomous transportation systems~\cite{yang2022semanticedge}, healthcare~\cite{ghodratnama2021summary2vec}, smart factories~\cite{luo2022semantic}, and so on.
In SemCom, we only transmit useful and necessary information to the recipients. The semantic extraction (SE) is a process wherein the useful and necessary features are extracted from the original raw data. For example, the essential speech features are extracted using an attention-based mechanism in~\cite{weng2021semantic,weng2021semantic2,tong2021federated}, image features are extracted using ResNet-50~\cite{he2016deep} in~\cite{zhang2022multi}, etc.  

{\color{black}In order to overcome the Shannon limit in 6G communication systems, the transmission data rate must be increased even further~\cite{yang2022semantic,qin2021semantic,strinati20216g}. The Shannon channel capacity may be exceeded for a communication system that transmits semantically correct data but permits a non-zero bit error rate (BER). A formal proof for the same is provided in~\cite{ma2023theory}. In a traditional communication system, the source information is transformed into bit sequences for processing. The bit sequence corresponding to the source information is precisely decoded at the receiver. In a traditional communication system, the bit/symbol transmission rate is limited by Shannon capacity. Semantic communication systems convey the semantic meaning of the source information. One significant distinction is the addition of semantic coding, which captures semantic information based on tasks or actions to be performed by the receiver. Only those semantic characteristics will be communicated, considerably reducing the number of essential communication resources. At the receiver, operations such as data reconstruction or more sophisticated tasks like image recognition and language translation might be performed. Using a semantic encoder with minimal semantic ambiguity and a semantic decoder with strong inference capabilities and a large shared knowledge base, it is proved in~\cite{bao2011towards} that there is a possibility to obtain higher transmission rates in semantic communication systems than those shown by Shannon's channel capacity theorem for traditional communication systems.}

During critical applications such as military operations, search operations by forest personnel in a dense forest, medical emergencies in remote areas, fire incidents in a remote agricultural land, the release of water from a nearby dam, etc., only the essential information needs to be communicated on an urgent basis. The messages could be in the form of text or audio and they come from a limited dataset. 
In a non-critical application, such as broadcasting a text/audio summary of commentary provided by live football commentators. Among all the words spoken by them, only a limited set of useful or important words are relevant to the game. These words are drawn from a limited dataset such as football vocabulary~\cite{footballvocab} which includes words such as \textit{goal, player names, red card, football, score, assist, half-time}, etc. This limited dataset provides an opportunity, in the context of SemCom design, for a significant overhead reduction by extracting and processing only the relevant keywords. For example, an uttered commentary sentence in 2022 FIFA world cup final game is:  `Messi shoots the ball into the right-bottom of the net and it's a goal!'  The extracted keywords in this example are \textit{Messi, shoots, ball, right-bottom, net, goal}. Only these keywords are transmitted in place of the entire sentence, and the receiver reconstructs a meaningful sentence. The reconstructed sentence in this case is: `Messi shoots the ball into the right-bottom of the net to score a goal.' This sentence is not exactly the same as the original sentence, but it conveys the same meaning.

The first goal of this paper is to use SemCom technology to reduce communication overhead, in the context of natural language processing (NLP) problems, while maintaining a certain minimum accuracy in wireless communication systems. The overhead reduction is performed with high accuracy in the literature~\cite{xie2021deep,xie2020lite}. However, in some applications, high data rates are preferred over high accuracy. {\color{black}For example, consider the football commentary transmission described in the preceding paragraph. Instead of sending the entire sentence, only the essential keywords are sent, and the entire sentence is reconstructed using the received keywords. The reconstructed sentence may not be completely accurate, but it conveys the same information as the original sentence with some accuracy. We were able to reduce overhead significantly while maintaining some accuracy. This reduction in overhead results in transmission of data with a faster rate, which is a key requirement for 6G.} As a result, we present the results of the trade-off between overhead reduction and accuracy. Model parameters are chosen based on the context.
Instead of transmitting raw data, the transmitter is designed to transmit semantic data, which significantly reduces network data traffic. A knowledge graph (KG) is a knowledge base that integrates data using a graph-structured topology. They are used to store interconnected event descriptions. These are used to predict the missing words in the received data (keywords) to construct a meaningful sentence.

{\color{black}Next, we apply the designed SemCom system to a realistic problem in which the transmitter and receiver are assumed to be a cloud server and a data center, respectively. A cloud server located on a remote cloud platform has access to a large raw dataset that is stored in the cloud. A data center is a centralized information storage facility that can store, process, and distribute massive amounts of data to its users~\cite{datacenter}. A data center requests a portion of the raw dataset in various categories. These dataset categories are based on their size and accuracy levels. For example, in the case of live streaming of sports events, the cloud server provides different quality videos of the same content to the data center, such as $480p, ~720p, ~1080p, ~2160p,$ and $4K$, and the users are served based on the subscribed video quality service and fee. The proposed SemCom technology-based communication system is used to transmit these datasets to the data center. Since the data center also has access to the shared KG, these datasets can be decoded with a certain accuracy. The data center replicates these different category datasets to store them in its storage facility to serve its subscribed users. Because storage capacity is limited, the data center's challenge is to find an optimal set of dataset replications to serve its subscribers.} The cloud server determines the price of each category dataset based on its size and quality (in terms of accuracy). 

In our case, the data center is assumed to store the different categories of the same portion of the dataset in its storage facility. The users can access these datasets directly from the storage facility by paying a certain price. The different-quality datasets are priced differently based on their quality and sizes. For example, a highly compressed dataset is small in size but poor in accuracy, so it is less expensive, whereas a lightly compressed dataset is large in size but superior in accuracy, so it is more expensive. Every user has a budget constraint, just as the data center has a storage constraint.\footnote{{\color{black} The assumption of data center's storage constraint is due to the following reasons: Data center storage, in general, refers to the devices, equipment, and software solutions that enable data storage within a data center facility. This includes data center storage policies and procedures that govern the entire data storage and retrieval process. Furthermore, data center storage may include data center storage security and access control techniques and methodologies. Data center components require substantial infrastructure to support such large hardware and software requirements. These include power subsystems, uninterruptible power supplies, ventilation, cooling systems, fire suppression, backup generators, and connections to external networks. To provide all of this necessary support, data centers typically have a limit on total data storage. Similarly, the constraints on the budget of its associated users are because each user has a limit on how much they can earn and, as a result, a limited budget that they can allocate to various items. Finally, limited income is the root cause of budget constraints. Budget constraints are visible in the fact that users cannot simply buy everything they want and are forced to choose between alternatives based on their preferences.}} The data center replicates these datasets based on the needs and budgets of the users. It strives to provide the highest-quality dataset possible to every user with a sufficient budget. It is a highly desirable scenario for the following two reasons: (a) the data center can maximize profits, and (b) the user is extremely satisfied with the service and can provide a high rating as feedback. Hence, we formulate an optimization problem in which the data center attempts to maximize its profit given the constraint on its storage capacity in order to serve all the subscribed users. We denote this problem as the \textit{Data Allocation Problem} (DAP).

{\color{black}However, challenges arise in data allocation due to resource constraints, such as data storage capacity at the data center and the need to serve all subscribed users. This difficulty is most noticeable when there are a large number of subscribers. In order to serve all subscribed users, the data center is forced to allocate lower-quality datasets to some of its subscribers, despite the fact that user feedback may damage its reputation and cost it revenue.}

{\color{black} Now, we provide the main contributions of this paper, which are as follows:
\begin{itemize}
    \item In this research work, we designed an auto-encoder-equipped transmitter and an auto-decoder-equipped receiver that only transmit the relevant keywords and, respectively, retrieve the data based on the received keywords and shared knowledge.
    \item We defined a new metric called Semantic Score (SS) that combines the best of two separate quantities, the BLEU score~\cite{papineni2002bleu}, and sentence similarity that employs BERT~\cite{devlin2019bert}, to quantify the overall semantic loss between the original and reconstructed sentences at the receiver.
    \item The performances of our proposed scheme are analytically compared to those of the two state-of-the-art schemes in terms of accuracy versus overhead reduction trade-off. 
    \item The SemCom system is then implemented in a realistic scenario in which a cloud server and a data center serve as transmitter and receiver, respectively. We formulated the DAP in which the data center optimally distributes different types of datasets received from the cloud server to its subscribers. We proved that the DAP is an NP-complete problem and proposed a greedy algorithm for solving it.
    \item We demonstrated using simulations on the real-world dataset that the proposed methods for SemCom system design outperform state-of-the-art methods in terms of average number of words per sentence for a given accuracy and that the proposed greedy algorithm solution of the DAP performs significantly close to the optimal solution.
\end{itemize}
}

The organization of the paper is as follows: A brief literature review of SemCom technologies and KGs is provided in Section~\ref{Sec:RelatedWork}. We introduce our SemCom system model and problem formulation in Section~\ref{Sec:SysModel_ProbForm}. The proposed SemCom system model is presented in detail in Section~\ref{Sec:Prop_SemCom_Model}. The performance analysis of the proposed scheme in terms of accuracy versus overhead reduction trade-off and cost comparisons is provided in Section~\ref{Sec:Performance_Analysis}. Then we define our DAP, show that it is an NP-complete problem, and propose a greedy algorithm to solve it in Section~\ref{Sec:DataAlloc}. We provide simulation results related to the proposed SemCom system model and the solutions of the DAP in Section~\ref{Sec:Simulations}. Finally, we conclude the paper along with a few future research directions in Section~\ref{Sec:Conclusions}.

\section{Related Work} \label{Sec:RelatedWork}
The study on SemCom technologies started very recently. The following state-of-the-art survey papers provide in-depth discussions on various SemCom technologies and their applications~\cite{strinati20216g,qin2021semantic,lan2021semantic,gunduz2022beyond,yang2022semantic,chaccour2022less}. Deep learning based SemCom technologies are proposed in ~\cite{xie2021deep,xie2020lite}. A brief tutorial on the
framework of SemCom and a method to calculate a bound on semantic data compression is provided in~\cite{niu2022towards}. The SemCom technology wherein both transmitter and receiver are empowered with the capability of contextual reasoning is proposed in~\cite{seo2021semantics}. The SemCom technology for a system where transmitter and receiver speak different languages is designed in~\cite{sana2022learning}. A {\color{black}multi-user} task-oriented SemCom system for multi-modal data transmission is proposed in~\cite{xie2021task}. A nonlinear transform based source-channel coding approach for SemCom is proposed in~\cite{9791398}, wherein a nonlinear transform mechanism is used to extract the source semantic features. 
A joint source-channel coding scheme is proposed in~\cite{farsad2018deep} that preserves the meaning between the transmitted sentence $s$ and the recovered sentence $\hat{s}$, while the two sentences may have different words and different lengths.
The work in~\cite{lu2022rethinking} introduced a reinforcement learning (RL)-powered SemCom paradigm that gives a system the ability to express semantics. 
In~\cite{wang2022performance}, a SemCom framework for textual data transmission is proposed. In this framework, semantic information is represented by a KG made up of a set of semantic triples and the receiver recovers the original text using a graph-to-text generation model. 
Another SemCom system based on the KG is proposed in~\cite{jiang2022reliable}. In this system, transmitted sentences are converted into triplets using the KG, which are seen as fundamental semantic symbols for semantic extraction and restoration, and they are ordered based on semantic relevance. 
All of these works are focused on achieving an overhead reduction without compromising the accuracy of the received data. None of these works investigated the possibility of further overhead reduction, thereby improving transmission data rate, while sacrificing a little accuracy. This issue is addressed in this paper using a shared knowledge base.  

A significant research on the usage of KGs is carried out in the field of natural language processing (NLP). The survey work in~\cite{hogan2021knowledge} provides a comprehensive study of KGs, which leverage large-scale data collections for usage in a variety of industry and academic applications. A survey paper based on KG is presented in~\cite{ji2021survey}. Similarly, another survey paper on KG text generation is presented in~\cite{yu2022survey}. A method to generate a summary of sentences by using a given set of keywords is proposed in~\cite{li2020keywords}. Similarly, a method to generate a summary of sentences by using a  knowledge base is shown in~\cite{huang2020knowledge}. Recently, KGs are {\color{black}utilized} in the context of SemCom design~\cite{wang2022performance,jiang2022reliable,zhou2022cognitive,liang2022life}. But these works do not focus on the issue presented in this paper, which is to design a SemCom system with a significant overhead reduction with a little compromise on accuracy. 

{\color{black} Problems similar to the DAP are well studied in the literature. A few similar problems studied are file allocation problems (FAP) in distributed systems~\cite{chandy1976file,chu1969optimal}, data allocation in database systems (DADS)~\cite{apers1988data,kwok1996design,karimi2009new}, data allocation over multiple channels at broadcast servers (DABS)~\cite{yee2002efficient}, etc. The goal of FAP is to find the optimal way to allocate files in order to minimize the operating costs associated with the files while keeping the following constraints in mind: (a) the expected time to access each file is less than a given bound; and (b) the amount of storage required at each computer does not exceed the available storage capacity. The problem in DADS is similar to FAP but differs in the following aspects: First, the data objects to be allocated are unknown in advance; second, these data objects are accessible by schedules that comprise transmissions between these data objects to generate the result. The goal of the DABS is to optimally allocate data objects across the available channels in a broadcast system so that the overall cost (in terms of average expected delay) is minimized while ensuring that all nodes receive at least one copy of these data objects. 
However, the uniqueness of the DAP lies in its problem structure. The data objects are allocated in the DAP with the goal of maximizing the data center's profit while ensuring that the cost of each allocated data object does not exceed the user budget and that each user is assigned exactly one category of data object.
}

\section{System Model and Problem Formulation}\label{Sec:SysModel_ProbForm}
\begin{figure*}
\centering
\includegraphics[width=1.0\textwidth]{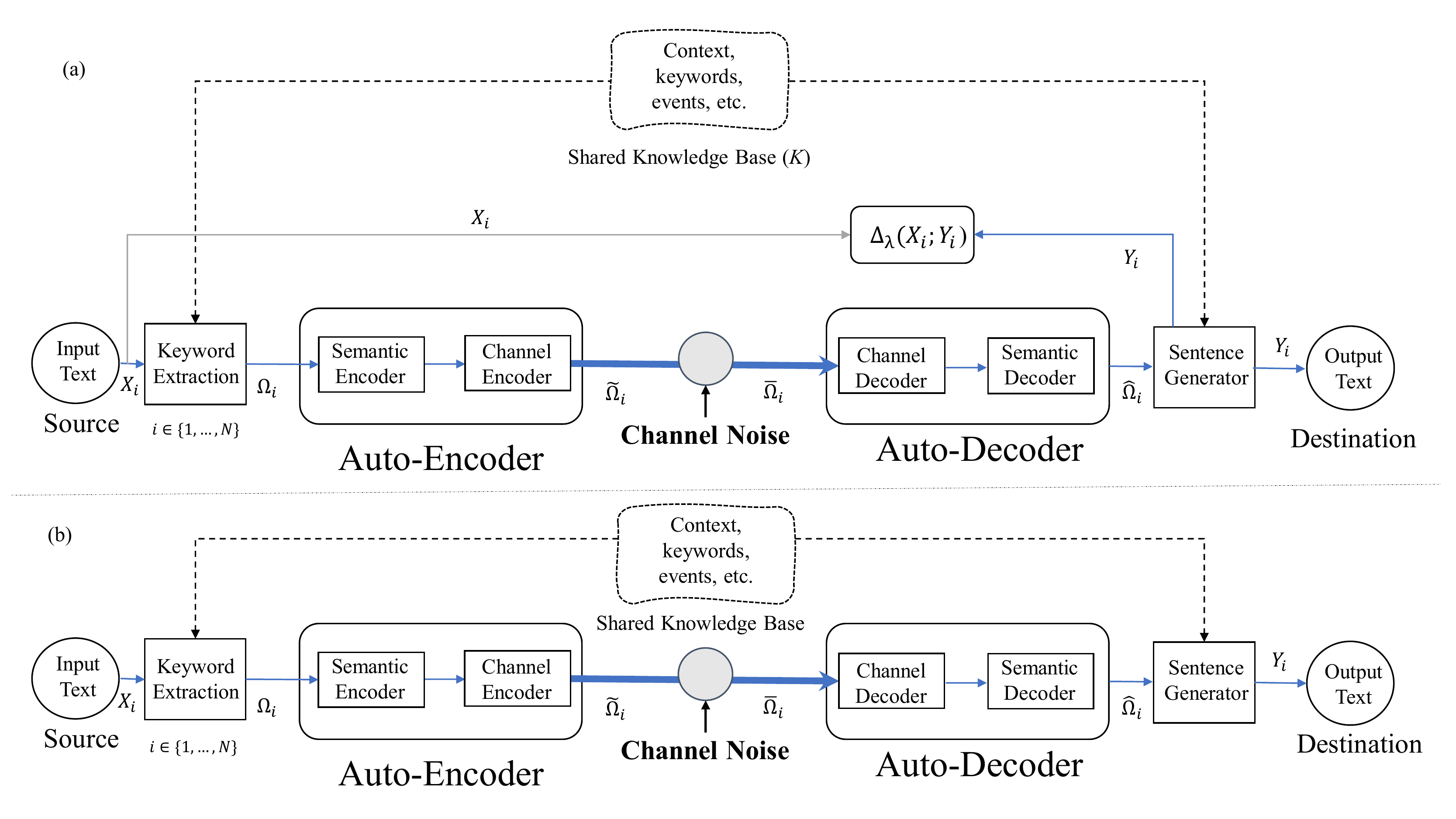}
    \caption{The block diagram of our proposed SemCom system model. The model in Fig. 1(a) is used for training the system parameters and the model in Fig. 1(b) is used for evaluating the SemCom system.}
\label{fig:SemComDesignModel}
\end{figure*} 
In this Section, first, we provide a brief overview of the proposed system model and later in Section~\ref{Sec:Prob_Formulation}, we present our problem formulation. {\color{black}In Table~\ref{Tab:MathSymbols}, we list the symbols and abbreviations used in this paper along with their meanings.}
\begin{table}[!ht]
    \centering
    \caption{List of mathematical symbols and abbreviations}
    \begin{tabular}{|c|l|}
        \hline
        \textit{Symbol} & \textit{Meaning} \\ 
        \hline
        SemCom & Semantic Communication \\ \hline
        DAP & Data Allocation Problem \\ \hline
        KG & Knowledge Graph \\ \hline
         $\mathcal{X}$ & Input text dataset \\ \hline
         $N$ & Number of sentences in $\mathcal{X}$
         \\ \hline
         $X_i$ & $i^{th}, i\in \{1, \ldots, N\}$, sentence of $\mathcal{X}$
         \\ \hline
         $\mathcal{K}$ & Shared knowledge base
         \\ \hline
         $\Omega_i$ & The set of keywords present in $X_i$
         \\ \hline
         $\Omega$& Total set of keywords \\ \hline 
         $\mathscr{S}_{\theta_e}$ & Semantic encoder  with $\theta_e$ as the parameter set
         \\ \hline
         $\mathscr{C}_{\phi_e}$ & Channel encoder with $\phi_e$ as the parameter set \\ \hline
         $\widetilde{\Omega}$ & The encoded set of symbols
         \\ \hline
         $h$ & Channel gain \\ \hline 
         AWGN & Additive White Gaussian Noise \\ \hline
        $\eta$ & AWGN noise
         \\ \hline
         $\overline{\Omega}$ & The set of received symbols at the receiver\\ \hline
         $\mathscr{C}_{\phi_d}$ & Channel decoder with $\phi_d$ as the parameter set \\ \hline $\mathscr{S}_{\theta_d}$ & Semantic decoder with $\theta_d$ as the parameter set \\ \hline
         $\widehat{\Omega}$ & Decoded set of keywords \\ \hline
         $Y$ & Set of
         sentences generated at the receiver \\ \hline 
         $\mathcal{Y}$ & Reconstructed text dataset \\ \hline
         SS & Semantic Score \\ \hline
         $\lambda$ & a hyper-parameter between 0 and 1 \\ \hline
          $\tau$ & user defined accuracy parameter \\ \hline
         BLEU & Bilingual Evaluation Understudy \\ \hline
         BERT & Bidirectional Encoder Representations from Transformers \\ \hline
          $\Delta_\lambda(s,\hat{s})$ & Semantic Score between sentences $s$ and $\hat{s}$
          \\ \hline BLEU$(s,\hat{s})$ & BLEU score between sentences $s$ and $\hat{s}$ \\ \hline
          $\Phi(s,\hat{s})$ & sentence similarity score between sentences $s$ and $\hat{s}$ \\ \hline
          $B$ & Batch size \\ \hline 
          $L$ & Total length of sentences \\ \hline
          $\textbf{M}$ & Semantic information provided by $\bf{\Omega}$ \\ \hline
          $V$ & Dimension of the encoder's output \\ \hline
          $\widehat{\textbf{M}}$ &
          Recovered semantic information \\ \hline
          $G$ & Number of data categories in data center \\ \hline
          $z_i$ & Size of $i^{th}$ category data in data center\\ \hline
          $c_i$ & Selling cost of of $i^{th}$ category data from data center\\ \hline
          $J$ & Number of subscribed users \\ \hline
          $Z$ & Total size constraint of data center \\ \hline
          $b_j$ & Budget constraint of $j^{th}$ user \\ \hline
          $d(z_i)$ & Purchase price of $i^{th}$ category data \\ \hline
          KP & Knapsack Problem \\ \hline
          $\overline{W}$ & Average number of words per sentence \\ \hline
          DeepSC & Deep learning enabled Semantic Communication \\ \hline
          JSCC & Joint Source-Channel Coding \\ \hline
          \end{tabular}    
    \label{Tab:MathSymbols}
\end{table}
\subsection{System Model} \label{Sec:SysModel}
The system model of the proposed SemCom system is shown in Fig.~\ref{fig:SemComDesignModel}. Let $\mathcal{X}$ be the input text dataset with $N$ sentences, $X_i$ be the $i^{th}, i\in \{1, \ldots, N\}$, sentence of $\mathcal{X}$, and $\mathcal{K}$ be the shared knowledge base. First, we extract the keywords from $\mathcal{X}$ using $\mathcal{K}$. Let the total set of keywords be $\Omega = \bigcup_{i=1}^N \Omega_i$, where $\Omega_i$ denotes the set of keywords present in $X_i$. The keyword extraction process is executed by multiplying the input sentence $X_i = [\omega_{i\ell}, \ell = 1, \ldots, L_i],$ with a binary vector $b_i = [b_{i\ell}, \ell = 1, \ldots, L_i]$, where $L_i = |X_i|$,\footnote{$|\mathcal{A}|$ denotes the cardinality of set $\mathcal{A}$.} which is defined as follows:
\begin{align}
b_{i\ell} &={
\begin{cases} {1,}& {\text{if  $\omega_{i\ell}$ is a keyword in $\mathcal{K}$}} \\ 0,& \text{else}.
\end{cases}}
\end{align}
Hence, $\Omega_i$, $i\in \{1, \ldots, N\}$, is obtained by collecting the non-zero elements from $X_i \odot b_i$, where $\odot$ is a word-wise multiplication operator. Here $X_i \odot b_i \triangleq [\omega_{i\ell} b_{i\ell}, \forall \ell = \{1, \ldots, L_i\}]$.\footnote{For ease of understanding, let us consider the example discussed in Section~\ref{Sec:Intro}. Let $X_i$ be [Messi shoots the ball into the right-bottom of the net and it's a goal!]. If the set of keywords present in $X_i$ is \{\textit{Messi, shoots, ball, right-bottom, net, goal}\} then $b_i = [1 1 0 1 0 0 1 0 0 1 0 0 0 1]$. Now, $X_i \odot b_i$ gives $[\textit{Messi, shoots, 0, ball, 0, 0, right-bottom, 0, 0, net, 0, 0, 0, goal}]$. Next, $\Omega_i$ is obtained by collecting the non-zero elements, i.e., $\Omega_i = \{\textit{Messi, shoots, ball, right-bottom, net, goal}\}$.} 
It is a simple search method, like checking whether the given English word is in the dictionary or not. For a faster and automatic keyword extraction process, refer to Rapid Automatic Keyword Extraction (RAKE) method~\cite{rose2010automatic}. It is an approach for extracting keywords from particular documents that is unsupervised, domain-independent, and language-independent.

Next, the $i^{th}$ keyword set $\Omega_i$ is encoded using the auto-encoder which consists of semantic and channel encoders. Let us denote $\mathscr{S}_{\theta_e}$ and $\mathscr{C}_{\phi_e}$ as the semantic and channel encoders with $\theta_e$ and $\phi_e$ as the parameter sets, respectively. After encoding $\Omega_i$, we get the following set of symbols:
\begin{equation}
    \widetilde{\Omega}_i = \mathscr{C}_{\phi_e} (\mathscr{S}_{\theta_e}(\Omega_i)),~i\in \{1, \ldots, N\}. 
\end{equation}

The encoded set of symbols $\widetilde{\Omega}_i$ is transmitted via the AWGN (additive white {\color{black}Gaussian} noise) channel. Let $h$ be the channel gain and $\eta$ be the noise which gets added to $\widetilde{\Omega}_i$ during transmission. So, the set of received symbols at the receiver is $\overline{\Omega}_i = h\widetilde{\Omega}_i + \eta$. After receiving, this set of symbols is decoded using the auto-decoder which consists of channel and semantic decoders. Let us denote $\mathscr{C}_{\phi_d}$ and $\mathscr{S}_{\theta_d}$ as the channel and semantic decoders with $\phi_d$ and $\theta_d$ as the parameter sets, respectively. After decoding $\overline{\Omega}_i$, we get the following set of keywords:
\begin{equation}
    \widehat{\Omega}_i = \mathscr{S}_{\theta_d} (\mathscr{C}_{\phi_d}(\overline{\Omega}_i)),~i\in \{1, \ldots, N\}.
\end{equation}

From the decoded set of keywords and the shared knowledge base $\mathcal{K}$, the sentence generator at the receiver generates the sentence $Y_i \in \mathcal{Y}, i\in \{1, \ldots, N\}$, where $\mathcal{Y}$ is the reconstructed text dataset.

\subsection{Problem Formulation}\label{Sec:Prob_Formulation}
Given the limited size of knowledge base, though the accuracy of the reconstructed sentences in $\mathcal{Y}$ may not be sufficiently high, the useful content in those sentences is summarized and conveyed to the receiver. This novel approach saves a significant amount of overhead. 

To measure the overall semantic loss between the original sentence $s \in \mathcal{X}$ and the reconstructed sentence $\hat{s} \in \mathcal{Y}$ at the receiver, we define a new metric named \textit{Semantic Score} (SS) which combines the best of two different quantities, viz., BLEU score (bilingual evaluation understudy~\cite{papineni2002bleu}) and sentence similarity which uses BERT~\cite{devlin2019bert}.\footnote{The detailed explanation on SS is provided in Section~\ref{Sec:SemanticScore}.}
Let $\Delta_\lambda(s,\hat{s})$ denote the SS between sentences $s$ and $\hat{s}$, which is a convex combination of $\text{BLEU}(s,\hat{s})$\footnote{When explicitly not mentioned then {\color{black}BLEU} 1-gram is used.} (see~\eqref{eq:BLEUScore}) and $\Phi(s,\hat{s})$ (see~\eqref{eq:SS_Score}), i.e.,
\begin{equation}
    \Delta_\lambda(s;\hat{s}) = (1-\lambda) \text{BLEU}(s,\hat{s}) + \lambda \Phi(s,\hat{s}),
    \label{eq:SemanticScore}
\end{equation}
where $\lambda \in [0,1]$ is a parameter. 
Note that pure $\text{BLEU}(s,\hat{s})$ and $\Phi(s,\hat{s})$ are obtained by setting $\lambda = 0$ and $\lambda=1$, respectively.

There exists a trade-off between overhead reduction and the accuracy that depends on the size of the knowledge base $\mathcal{K}$. For example, if the size of the set $\mathcal{K}$ is small, then on an average only a few keywords are extracted from the given input sentences in $\mathcal{X}$, encoded and transmitted, which implies higher amount of average overhead reduction. This creates a large amount of missing information on an average, due to which accuracy of the reconstructed sentences in $\mathcal{Y}$ is expected to be low. 
On the other side,  if the size of the set $\mathcal{K}$ is large, then on an average a significant number of keywords are extracted from the given sentences in $\mathcal{X}$, encoded and transmitted, which implies lower average overhead reduction. This creates a small amount of missing information on an average, due to which accuracy of the reconstructed sentences in  $\mathcal{Y}$ is expected to be high. This phenomenon is numerically shown in Section~\ref{SubSec:Simu_SemCom}.

So, in this paper we aim at minimizing the transmission of average number of words per sentence (equivalent to maximizing the average overhead reduction) by keeping a certain minimum accuracy information $\tau$ in the received sentence, i.e.,
\begin{subequations}
\begin{align}
    \min & ~ \frac{1}{N}\sum_{i=1}^N  |\Omega_i| \label{eq:min_words} \\
     &\Delta_\lambda(X_i;Y_i) \ge \tau,~i=\{1, \ldots, N\}, \label{eq:tau_thr} 
\end{align}
\end{subequations}
where $|\Omega_i|$ denotes the number of keywords in $\Omega_i$ that corresponds to sentence $X_i$, and $\tau$ is an user defined parameter. 


\subsubsection{Semantic Score (SS)}\label{Sec:SemanticScore}
Now, we describe the semantic score used in the problem formulation. An attempt was made earlier in ~\cite{malandrakis2012deeppurple} where features of BLEU score and sentence similarity score are integrated using a multiple linear regression model. It determines sentence lexical matches, lexical semantic similarity between non-matching words, and sentence lengths. Motivated by this work, we propose the idea of SS (see~\eqref{eq:SemanticScore}). The BLEU score cannot handle word synonyms, but it is a fast, low-cost algorithm that is language-independent and corresponds with human judgment. The sentence similarity score using BERT vectors is slow, has comparable ratings to the BLEU, but it also handles synonyms. A brief description of the BLEU score and sentence similarity score is provided in the next two paragraphs.

First, let us define the quantity BLEU score (bilingual evaluation understudy~\cite{papineni2002bleu}) to compare the similarities between two sentences quantitatively. 
The BLEU$(s,\hat{s}) \in [0,1]$ score between transmitted sentence $s$ and reconstructed sentence $\hat{s}$ is computed as follows:
\begin{equation}
    \text{BLEU}(s,\hat{s}) = \text{BP}(s,\hat{s}) \exp\left( \sum_{n=1}^{W} {w_n \ln{p_n (s,\hat{s})}}\right),
    \label{eq:BLEUScore}
\end{equation}
where $p_n$ denotes the modified $n$-gram precision function up to length $W$, $w_n$ denotes the weights, and brevity penalty (BP) is given by the following expression:
\begin{align}
\text{BP}(s,\hat{s}) &={
\begin{cases} {1}& {\ell_c > \ell_r} \\ e^{1-\ell_r/\ell_c}& \ell_c \le \ell_r,
\end{cases}}
\end{align}
where $\ell_c$ is the length of the candidate translation and $\ell_r$ is the effective reference corpus length~\cite{papineni2002bleu}.

Next, sentence similarity score $\Phi(s,\hat{s})$ is defined as follows:
\begin{equation}
    \Phi(s,\hat{s}) \triangleq \frac{\boldsymbol{\beta(s)\cdot {\beta(\hat{s})}}^T}{||\boldsymbol{\beta(s)}||~|| \boldsymbol{\beta(\hat{s})}||},
    \label{eq:SS_Score}
\end{equation}
where $\boldsymbol{\beta}$, representing BERT~\cite{devlin2019bert}, is a massive pre-trained model, which uses word embeddings, with billions of parameters used to extract semantic information. The $\Phi(s,\hat{s})$ is a number between 0 and 1, indicating how similar the reconstructed sentence is to the transmitted sentence, with 1 indicating the highest similarity and 0 indicating no similarity between $s$ and $\hat{s}$. Word embeddings are vectors that have been mapped to words to assist computers in interpreting language. For example, the words \textit{cat} or \textit{dog} are difficult for a computer to understand, their vector form is more appropriate. One assumption of embedding mapping is that related words should be near each other. Note that word embedding vectors with contextualized embeddings have distinct embeddings for the same word depending on its context. Hence, these embeddings are proposed to be used as sentence-level embeddings. 

\subsubsection{Shared Knowledge Base}
In this paper, we generate the shared knowledge base $\mathcal{K}$ by using the keywords from a limited dataset $\Omega$ which consists of only the relevant words of a particular event, like that of a football game in our case.  We assume that both transmitter and receiver have access to $\mathcal{K}$. 
During the feature extraction process, in every sentence, the words $w \in \Omega$ are encoded into their corresponding symbols and transmitted to the receiver in their corresponding time slots. 
At other time slots, a common symbol is transmitted.
By utilizing $\mathcal{K}$, the receiver reconstructs the sentence based on the received words in that sentence. To improve the accuracy of the reconstructed sentences, we can increase the size of $\mathcal{K}$ by adding more keywords from the vocabulary generated using $\mathcal{X}$. 
This result is shown using simulations in Section~\ref{SubSec:Simu_SemCom}. 

\section{Proposed SemCom System Model} \label{Sec:Prop_SemCom_Model}
\begin{figure*}
\centering
\includegraphics[width=1.0\textwidth]{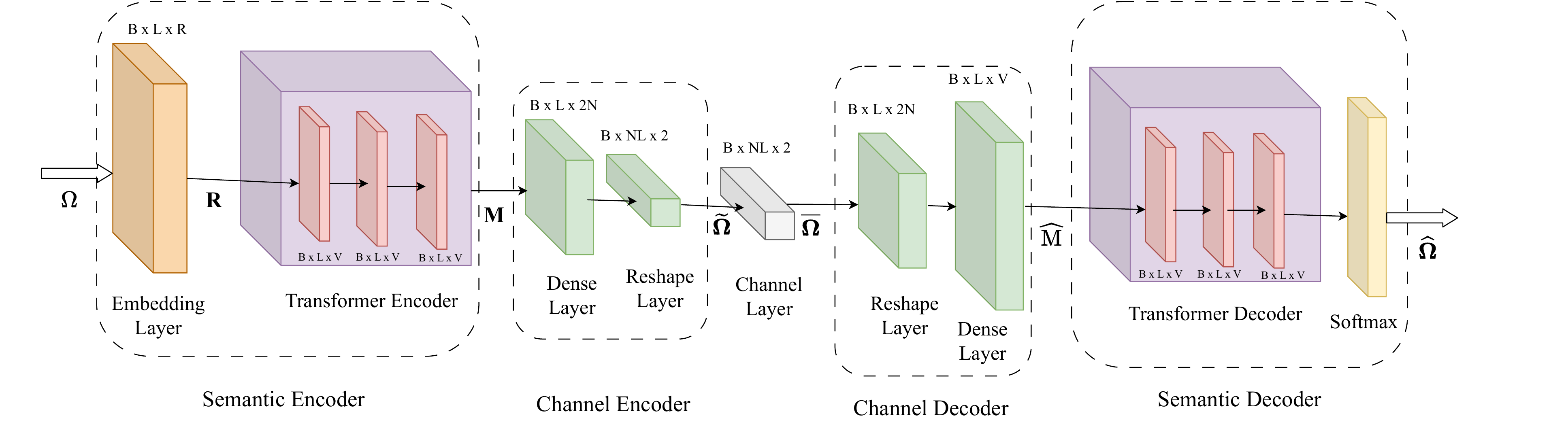}
    \caption{The architecture of the semantic encoder/decoder and channel encoder/decoder models of the proposed SemCom system model.}
\label{fig:SemComArchitecture}
\end{figure*} 
The detailed architecture of the semantic and channel encoder/decoder models is shown in Fig.~\ref{fig:SemComArchitecture}.
The transmitter comprises of a semantic encoder that extracts semantic characteristics from the texts to be broadcast and a channel encoder that generates symbols to assist further transmission. The conceptual encoder has many Transformer encoder layers~\cite{vaswani2017attention}, whereas the channel encoder has dense layers with various units. In the model, the AWGN channel is considered as one layer. As a result, the receiver is composed of a channel decoder for symbol identification and a semantic decoder for text estimation, with the channel decoder consisting of dense layers with varied units and the semantic decoder consisting of several Transformer decoder layers~\cite{vaswani2017attention}. 

Let $\epsilon \triangleq 1-SS$. During training, each step attempts to minimize $\epsilon$ using gradient descent with mini-batch until the stop condition is fulfilled, the maximum number of iterations is achieved, or none of the terms in the loss function are reduced anymore. Unlike separate semantic and channel coding, where the channel encoder/decoder deals with digital bits rather than semantic information, joint semantic-channel coding can maintain semantic information when compressing data~\cite{farsad2018deep}.

\subsection{Training the System Model}
The training of the SemCom model shown in Fig.~\ref{fig:SemComDesignModel}(a) is carried out keeping the SS in mind.
A {\color{black}mini-batch} of sentences $\textbf{S} \in \mathscr{R}^{B\times L}$, are converted to sets of keywords $\bf{\Omega}~$$\in \mathscr{R}^{B\times L}$, where $B$ and $L = \sum_{i=1}^B L_i$ are batch size and total length of the sentences, respectively, is transmitted to the semantic encoder. These sentences can be represented as a dense word vector $\textbf{R} \in \mathscr{R}^{B\times L \times R}$ which are obtained after passing through the embedding layer, where $R$ is the dimension of the word vector. This word vector $\textbf{R}$ is then passed to the Transformer encoder, primary component of Semantic Encoder, to acquire the semantic information $\textbf{M} \in \mathscr{R}^{B\times L \times V}$,  provided by $\bf{\Omega}$, where $V$ is the dimension of the encoder's output. Then, to account for the effects of the physical channel noise, \textbf{M} is encoded into symbols $\widetilde{\bf{\Omega}}$, where $\widetilde{\bf{\Omega}} \in \mathscr{R}^{B\times NL \times 2}$, which constitutes the channel encoder which is implemented using dense layer followed by reshape layer. Next, the receiver receives distorted symbols $\overline{\bf{\Omega}}$ after travelling through the channel. 
The channel decoder layer, implemented using reshape layer followed by dense layer, decodes distorted symbols $\overline{\bf{\Omega}}$ received at the receiver, where $\widehat{\textbf{M}} \in \mathscr{R}^{B\times L \times V}$ is the recovered semantic information of the sources. The semantic decoder layer then estimates the transmitted sentences $\widehat{\bf{\Omega}}$ with the help of sentence generator and the shared knowledge $\mathcal{K}$. Finally, the stochastic gradient descent
(SGD) method is used to optimize the network, using the error $\epsilon$.

\subsection{Performance Analysis} \label{Sec:Performance_Analysis}
Now, we analytically compare the performances of our proposed scheme with those of the DeepSC scheme~\cite{xie2021deep} and the adaptive scheme~\cite{sana2022learning}.  
\subsubsection{Accuracy versus Overhead Reduction Trade-off}
Let $X_i, i = \{1, \ldots, N\}$, denote the $i^{th}$ sentence in the dataset $\mathcal{X}$ and $\omega_{i\ell}, \ell=\{1, \ldots,|X_i|\}$, denote the $\ell^{th}$ word in the sentence $X_i$, then we can write,
\begin{equation}
    X_i = \{\omega_{i\ell} |~ \ell=1, \ldots,|X_i|\}, \forall i \in \{1, \ldots, N\}.
\end{equation} 
In our proposed scheme, as described in the system model (see Section~\ref{Sec:SysModel}), we only transmit the extracted keywords before encoding. Hence, the total number of words present in the dataset $N_0$ and the total number of keywords to be transmitted $N_\tau$, respectively, are
\begin{subequations}
\begin{align}
    N_0 &= \sum_{i=1}^N |X_i|, \\
    N_\tau &= \sum_{i=1}^N |\Omega_i(\tau)|,
\end{align}   
\end{subequations}
where $\Omega_i(\tau), i = \{1, \ldots, N\}$, is the set of keywords present in the sentence $X_i$ for a given accuracy $\tau$. 
Let $n_0$ denote the fixed number of symbols used to represent a word in the DeepSC scheme~\cite{xie2021deep}. So, the total number of symbols used for communicating the whole data in the DeepSC scheme, $\Psi_0$, and the proposed scheme, $\Psi_\tau$, respectively, are as follows:
\begin{subequations}
\begin{align}
    \Psi_0 &= n_0 N_0, \label{eq:SymbolsDeepSC}\\
    \Psi_\tau &= n_0 N_\tau. \label{eq:SymbolsProp}
\end{align}    
\end{subequations}
 Let $p_{i\ell}$ denote the probability of occurrence of the word $\omega_{i\ell}$, i.e.,
\begin{equation}
    p_{i\ell} = \frac{|\omega_{i\ell}|}{\sum_{i=1}^N \sum_{\ell=1}^{|X_i|}|\omega_{i\ell}|} , \ell \!\in \!\{1, \ldots, |X_i|\}, i \!\in \!\{1, \ldots, N\}.
\end{equation}
Now, based on the value $p_{i\ell}$, in the adaptive scheme the number of symbols used to encode the word $\omega_{i\ell}$ is chosen using the following equation:
\begin{equation}
    q_{i\ell} = \min\left(\max(n_{\min},\lfloor n_0 N_0 p_{i\ell}+0.5\rfloor),n_0 \right),
    \label{eq:q_il}
\end{equation}
where $n_{\min} < n_0$ denotes the minimum number of possible symbols. Hence, the total number of symbols used in the adaptive scheme is 
\begin{equation}
 \widehat{\Psi} = \sum_{i=1}^N \sum_{\ell=1}^{|X_i|} q_{i\ell}.   \label{eq:SymbolsAdaptive}
\end{equation}

Let $\alpha_\tau^d$, $\alpha_\tau^a$, and $\alpha_\tau^p$ represent the product of the average number of symbols used for each word and the fraction of total words transmitted required to achieve the accuracy $\tau$ in DeepSC~\cite{xie2021deep}, adaptive scheme~\cite{sana2022learning}, and proposed scheme, respectively. Since all words are transmitted in both DeepSC and adaptive schemes, the $\alpha_\tau$ values for these schemes are as follows: 
\begin{subequations}
\begin{align}
    \alpha_\tau^d &= n_0(\tau), \\
    \alpha_\tau^a &= \widehat{\Psi}(\tau)/N_0,
\end{align}    
\end{subequations}
where $n_0(\tau)$ and $\widehat{\Psi}(\tau)/N_0$ are the average number of symbols required to achieve accuracy $\tau$ in DeepSC and adaptive schemes, respectively. Note that $n_0(1) = n_0$ and $\widehat{\Psi}(1) = \widehat{\Psi}$.  
In case of the proposed scheme, only a fraction of all the words are transmitted. It uses the same number of symbols as that of the DeepSC scheme, that is $n_0(\tau)$, but it is still able to achieve better results due to the transmission of only the keywords. Hence, the value of $\alpha_\tau$ for the proposed scheme is
\begin{equation}
    \alpha_\tau^p = n_0(\tau)\frac{N_\tau}{N_0}.
\end{equation} 
The accuracy vs. overhead reduction {\color{black}trade-off} can be compared among the proposed scheme, DeepSC, and adaptive scheme by measuring the $\alpha_\tau$ values obtained by each of these schemes. 
We numerically compare these values, and the results are shown in Fig.~\ref{fig:Avg_Sym_plots}.

\subsubsection{Cost Comparisons}
Now, we analyse the average costs incurred for transmission of a sentence in various schemes. 
Let $t_\omega$ (in $\mu s$) be the time to check word $\omega$ whether it is a keyword or not from the knowledge base $\mathcal{K}$. Let the cost equivalent of spending time $t_\omega$ on such operation is $c(\omega)$. So the average cost spent on keyword extraction process can be found as
\begin{equation}
    C_k = \frac{1}{N}\sum_{i=1}^N \sum_{\ell=1}^{|X_i|}c(\omega_{i\ell}).
\end{equation}
Next, we compute the costs involved in the transmission process.\footnote{These costs include transmission power, encoding and decoding processes, etc.} Let $\bar{c}(\xi)$ denote the cost of spending resources for the transmission of symbol $\xi$. The total number of symbols used in the schemes DeepSC, proposed, and adaptive, respectively, can be computed by using~\eqref{eq:SymbolsDeepSC},~\eqref{eq:SymbolsProp}, and~\eqref{eq:SymbolsAdaptive}. So the average costs incurred for the proposed ($C_t^p$), DeepSC ($C_t^d$), and adaptive ($C_t^a$)  schemes, respectively, in the transmission process are
\begin{subequations}
    \begin{align}
        C_t^p &= \frac{n_0}{N}\sum_{i=1}^N \sum_{\ell=1}^{|\Omega_i|}\bar{c}(\xi_{i\ell}),\\
       C_t^d &= \frac{n_0}{N}\sum_{i=1}^N \sum_{\ell=1}^{|X_i|}\bar{c}(\xi_{i\ell}),\\   
       C_t^a &= \frac{1}{N}\sum_{i=1}^N \sum_{\ell=1}^{|X_i|}q_{i\ell}\bar{c}(\xi_{i\ell}).
    \end{align}
\end{subequations}
In the proposed scheme, after recovering the set of keywords, the receiver has to reconstruct the sentences for the meaningful recovery. Let $\hat{c}(x)$ be the cost of reconstructing sentence $\mathcal{X}$ at the receiver. So the average cost spent on sentence reconstruction process can be found as
\begin{equation}
    C_r = \frac{1}{N}\sum_{i=1}^N \hat{c}(x_i).
\end{equation}
The average costs incurred for end-to-end transmission of a sentence for DeepSC, adaptive, and proposed schemes are $C_t^d$, $C_t^a$, $C^p = C_k + C_t^p + C_r$, respectively. 
\begin{remark}
    For communication systems with large computing and storage capabilities, pre-processing and post-processing operations like keyword extraction and sentence generation, respectively, incur a marginal cost. Also, in the case of poor channel conditions and highly congested networks, the amount of information (symbols) to be transmitted becomes a crucial factor for efficient data communication. In these scenarios, our proposed scheme outperforms DeepSC~\cite{xie2021deep} and adaptive~\cite{sana2022learning} schemes in terms of average costs incurred, i.e.,  $C^p < C_t^a < C_t^d$.
\end{remark}

\section{Data Allocation Problem} \label{Sec:DataAlloc}
Let us assume that the transmitter and receiver shown in Fig.~\ref{fig:SemComDesignModel} are a cloud server and a data center, respectively. The original copy of the dataset $\mathcal{X}$ is stored in the cloud, and a data center requests a portion of it from the cloud server, say $\overline{X} \subset \mathcal{X}$. The cloud server uses SemCom technology to communicate the requested portion of data to the data center, as described in Section~\ref{Sec:SysModel}. The data center obtains each copy of $\overline{X}_\tau$, for some specific values of $\tau \in [0,1]$, by tuning the parameter $\tau$. It can only make a limited number of copies of these data due to storage constraints. This data center serves a number of users, each of whom has a budget constraint. Assume that the users do not have sufficient memory to store the data. They use data stored in the data center. Once data portions are allocated, users can access them whenever they need. Next, we formulate an optimization problem, which we call \textit{Data Allocation Problem} (DAP), to maximize the profit for this data center given its storage constraint and the budget constraints of its associated users. 

\begin{figure}
\centering
\resizebox{0.94\columnwidth}{!}
{\includegraphics{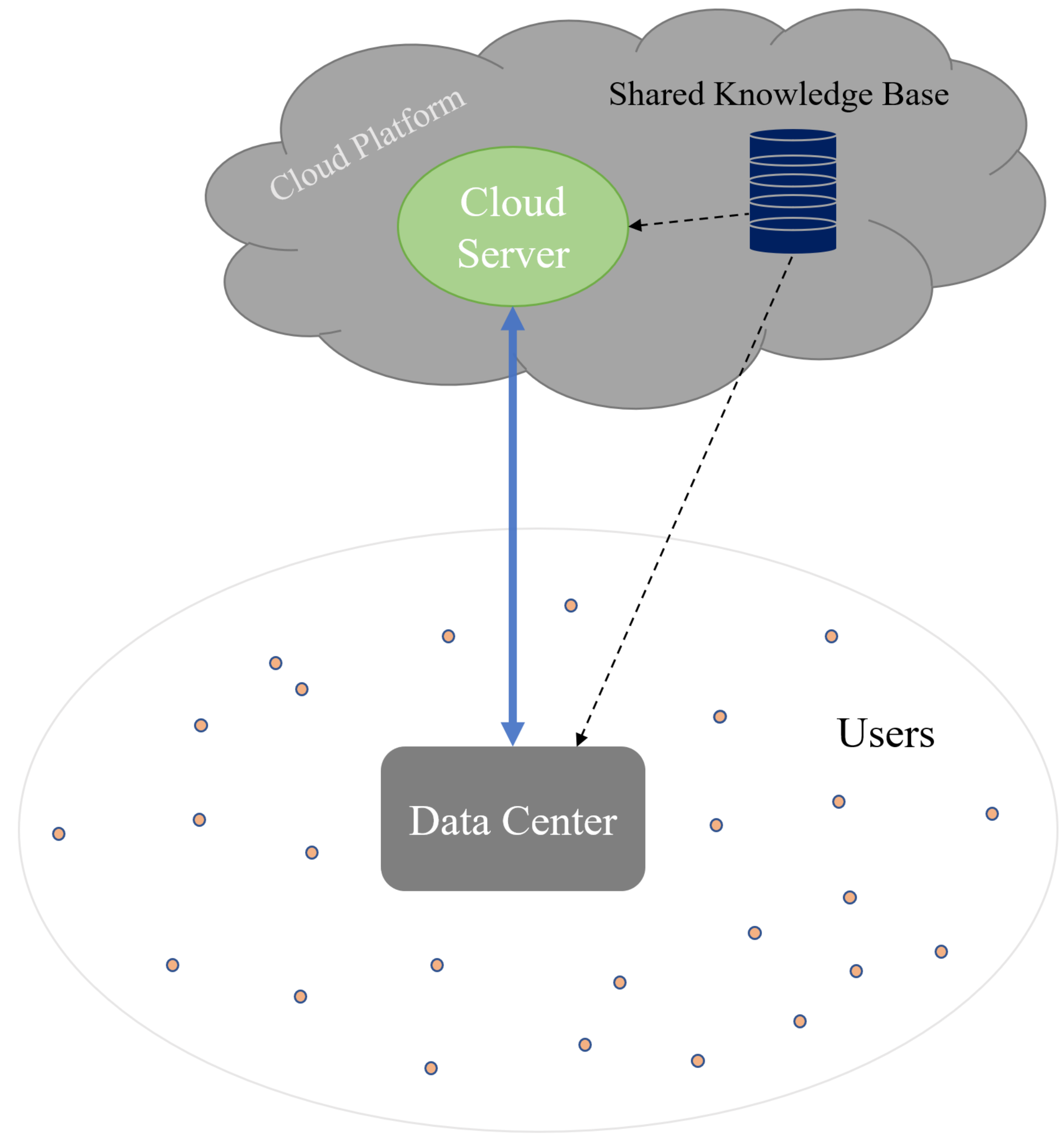}}
\caption{This figure shows the setup used to describe the data allocation problem (DAP).}
\label{fig:Cloud_DataCenter_Network}
\end{figure}

The setup used to describe the DAP is shown in Fig.~\ref{fig:Cloud_DataCenter_Network}.\footnote{{\color{black}Real-world examples of the setup shown in Fig.~\ref{fig:Cloud_DataCenter_Network} are as follows: (a) Microsoft Azure provides storage facilities for the subscribed users in its data center~\cite{AzureStorage}.  The storage options include premium and standard quality options to optimize costs and workload performance. Microsoft Azure guarantees sub-millisecond latency in its storage facility for high throughput and transaction-intensive workloads. The prices quoted for using the storage facility are determined by the users' location ~\cite{AzurePricing}. For example, prices per month of usage for users in the Korean Central Region are $\$0.81, ~\$1.62, ~\$3.24, ~\$5.28, ~\ldots,$ with storage sizes of $4GB, ~8GB, ~16GB, ~32GB, ~\ldots,$ respectively. (b) Amazon Web Services offers its subscribers an easy-to-use, scalable, high-performance block-storage service called Elastic Block Store (EBS)~\cite{AmazonStorage}. The quoted prices for using the EBS facility are decided by the user's location~\cite{AmazonPricing}. For example, for users in the Asia-Pacific (Seoul) region, the prices per month of usage are $\$0.37, ~\$0.72, ~\$1.46, ~\$2.92, ~\ldots,$ for data storage sizes of $4GB, ~8GB, ~16GB, ~32GB, ~\ldots,$ respectively.}} Let $G$ be the number of data categories that the data center has based on the accuracy $\tau$. They are indexed by $i=1, \ldots, G$, such that $\tau \in [\tau_{\min}, \tau_{\max}], \tau_{\min} < \tau_{\max}$, is one-to-one corresponding to $i \in \{1, \ldots, G\}$. The data center can produce $m_i$ copies of $\overline{X}_i, i=1, \ldots, G$. Each copy of data is $z_i$ in size, and its selling cost is $c_i$ and they are related by $z_{i_1} < z_{i_2}$ and $c_{i_1} < c_{i_2}$ for $i_1 < i_2$, $\forall i_1 \in \{1, \ldots, G-1\}, \forall i_2 \in \{2, \ldots, G\}$. These constraints indicate that the sizes and  costs of different categories of data increase as $\tau$ increases, which is shown using the indices $i \in \{1, \ldots, G\}$. Now, we would carefully incorporate the results obtained from~\eqref{eq:min_words} and~\eqref{eq:tau_thr} in terms of size and accuracy, so as to make the DAP meaningful from the perspective of SemCom developed in this paper.

Assume the data center serves $J$ users, with each user having a budget constraint $b_j \ge c_1, \forall j \in \{1, \ldots, J\}$.\footnote{This constraint ensures that every user is eligible to access at least one category of data.} 
Let $U_{i,j}$ represent an indicator variable that returns 1 when $b_j \ge c_i$, which means that the $i^{th}$ category data can be provided to user $j$ when its cost is not more than the user budget and 0 otherwise, i.e., for a given $i \in \{1, \ldots, G\}, j \in \{1, \ldots, J\}$,
\begin{equation}
U_{i,j} =
    \begin{cases}
     1, ~~ \text{if $b_j \ge c_i$,}  \\
     0,~~ \text{else.}
    \end{cases} \label{eq:u_ij}
\end{equation}
Similarly, let $V_{i,j}$ represent an indicator random variable that returns 1 when the $i^{th}$ category data is actually provided to user $j$ and 0 otherwise, i.e.,
\begin{equation}
V_{i,j} =
    \begin{cases}
     1, \text{if category  data $i$ is actually provided to user $j$,}  \\
     0,\text{else.}
    \end{cases}
\end{equation} 
So, the value of $m_i$ can be obtained as follows:
\begin{equation}
    m_i = \sum_{j=1}^J  V_{i,j}, ~\forall i \in \{1, \ldots, G\}. \label{eq:mi_Vij}
\end{equation}
The value of $m_i$ computed using~\eqref{eq:mi_Vij} shows that it also denotes the number of users who are provided with $i^{th}, i \in \{1, \ldots, G\}$, category data. 

Now we consider the purchase price of the data from the cloud server. The limited backhaul capacity between the cloud server and data center constrains the rate of data transfer between them. Due to this, the size of the data, and hence the use of SemCom technology, plays an important role. The purchase prices $d(z_i), i \in \{1, \ldots, G\},$ of different categories of data from the cloud server are based on their sizes, i.e., $d(z_{i_1}) < d(z_{i_2})$  for $i_1 < i_2$, $\forall i_1 \in \{1, \ldots, G-1\}, \forall i_2 \in \{2, \ldots, G\}$. 
Based on this information, an optimization problem, which we call DAP, to maximize the profit of the data center is formulated as follows:
 
\begin{subequations}
\begin{align}
    \max_{\substack{{V_{i,j},}\\{i \in \{1, \ldots, G\},}\\{j \in \{1, \ldots, J\}}}} & ~ \sum_{i=1}^G \left(c_i\sum_{j=1}^J  V_{i,j} - d(z_i) \right)\label{eq:max_profit} \\
     &\sum_{i=1}^G \left(z_i\sum_{j=1}^J  V_{i,j}\right) \le Z, \label{eq:size_constraint} \\
     & \sum_{i=1}^G \sum_{j=1}^J U_{i,j} V_{i,j} = J, \label{eq:category_constraint}
     \\
     & \sum_{i=1}^G V_{i,j} = 1,~j \in \{1, \ldots, J\}, \label{eq:user_constraint}
     \\
     & V_{i,j} \in \{0,1\},\forall i\in \{1, \ldots, G\}, j \in \{1, \ldots, J\}. \label{eq:integer_constraint}
\end{align}
\end{subequations}
The constraint~\eqref{eq:size_constraint} ensures that the total size of all copies of allocated data is within the limit of the maximum permissible size $Z$ at the data center. 
Similarly, the constraint~\eqref{eq:category_constraint} indicates that all data portions are assigned to users while making sure that the cost of every allocated data portion is not more than the user budget (see~\eqref{eq:u_ij}), whereas the constraint~\eqref{eq:user_constraint} indicates that each user $j \in \{1, \ldots, J\}$ is allocated exactly one category of data. The last constraint~\eqref{eq:integer_constraint} indicates that the variables $V_{i,j},\forall i\! \in \!\{1, \ldots, G\}, j\! \in \!\{1, \ldots, J\}$, are binary. 
This makes our optimization problem, DAP, defined in~\eqref{eq:max_profit}-\eqref{eq:integer_constraint} as a type of binary integer programming. 

In general, the integer programming problems are shown to be NP-complete~\cite{karp1972reducibility}. We show that the DAP belongs to the class of NP-complete problems by reducing the well known knapsack problem (KP)~\cite{martello1990knapsack} to it. 

\begin{theorem} \label{Thm_DAP}
The DAP is NP-complete.
\end{theorem}
\begin{proof}
The proof is given in Appendix~\ref{Apdx_Thm_DAP}
\end{proof}

\subsection{Greedy Algorithm}\label{sec:greedyalgo}
In Theorem~\ref{Thm_DAP}, we have shown that the DAP belongs to the class of NP-complete problems. Now, we present a greedy algorithm to solve the DAP. First, we identify the condition under which the solution is feasible. From the DAP formulation, it is clear that there is a limit to the number of users that the data center can serve. In the worst-case scenario, all users could be assigned the least desirable data category, $i=1$. The total data size in this case is $z_1$ times $J$ and is limited by $Z$. Hence the condition for the solution to exist is:
\begin{equation}
    J \le \frac{Z}{z_1}.
\end{equation}

Given that the primary goal of the DAP is to maximize profit for the data center while ensuring data allocations to all users, the proposed algorithm allocates the best possible category data to each user based on their budget. This is accomplished by determining $k(j) \in \{1, \ldots, G\}, \forall j \in \{1, \ldots, J\},$ such that $c_{k(j)} \le b_j < c_{k(j)+1}$ (which is same as finding $i$ such that $U_{i,j} = 1$ and $U_{i+1,j} = 0$), and allocating the data category $i=k(j)$ for $j^{th}$ user. This gives $V_{i,j} = V_{k(j),j} =1, \forall j \in \{1, \ldots, J\}$. Next, the algorithm computes the total size due to this allocation policy, i.e., $\overline{Z} = \sum_{j=1}^J z_{k(j)} V_{k(j),j}$. If $\overline{Z} \le Z$, then we have found the solution, $V^{\star}$, of the DAP and it is as follows:
\begin{equation}
V^{\star}_{i,j} =
    \begin{cases}
     1, ~~ \text{if $i = k(j)$,}  \\
     0,~~ \text{else.}
    \end{cases} \label{eq:v_ij_solution}
\end{equation}
And the profit is $\mathrm{P} = \sum_{i=1}^G \left(c_i\sum_{j=1}^J  V^{\star}_{i,j} - d(z_i) \right)$. But if $\overline{Z} > Z$, which implies the violation of the constraint~\eqref{eq:size_constraint}, then the algorithm updates the data allocation policy in the following way. It finds the smallest argument $k(j')$ which minimizes the ratio $r_{k(j)} = (c_{k(j)}/z_{k(j)}), \forall j \in \{1, \ldots, J\}$, and does the following updates using it: $V_{k(j'),j'} = 0$, $V_{k(j')-1,j'} = 1$,
$\overline{Z} \rightarrow \overline{Z} - z_{k(j')} + z_{k(j')-1}$, $k(j') \rightarrow k(j')-1$.\footnote{This approach ensures the smallest possible reduction in selling cost from $\mathrm{P}$ to $(\mathrm{P} - c_{k(j')} + c_{k(j')-1})$ and/or the largest possible data size reduction from $\overline{Z}$ to $(\overline{Z} - z_{k(j')} + z_{k(j')-1})$, which aids in satisfying the constraint~\eqref{eq:size_constraint}. If only the selling cost is considered in place of the ratio, which is the case in most greedy algorithms, the algorithm ignores the impact of data sizes on the DAP. We call this algorithm as \textit{greedy-cost} algorithm and show, using simulations, in Section~\ref{SubSec:Simu_DAP} that the proposed greedy algorithm outperforms the greedy-cost algorithm.} The algorithm again compares $\overline{Z}$, computed with updated value of $k(j')$, and $Z$. This process continues until it encounters  $\overline{Z} \le Z$ and the solution is $V^\star = V$. The detailed algorithm is provided in Algorithm~\ref{alg:greedy}.

\begin{algorithm}
\caption{Greedy Algorithm}\label{alg:greedy}
\begin{algorithmic}[1]
\State \textbf{Input:} $c_i, z_i, d(z_i), U_{i,j}, i \in \{1, \ldots, G\}$, $j \in \{1, \ldots, J\}$, $G, J, Z$
\If {$J \le Z/z(1)$}
\State Initialize $j=1$, $r(1) = \infty$, $i=2$, $V = \textbf{0}_{G \times J}$.
\While {$i \le G$} 
\State $r_i \gets c_{i}/z_{i}$, $i \gets i+1$.
\EndWhile
\While {$j \le J$} 
\State Find $i$ such that $U_{i,j} = 1$ and $U_{i+1,j} = 0$. 
\State $k(j) \gets i$, $V_{k(j),j} \gets 1$, $j \gets j+1$.
\EndWhile
\State Compute $\overline{Z} = \sum_{j=1}^J z_{k(j)}$ (Note: $ V_{k(j),j}=1, \forall j \in \{1, \ldots, J\}$).
\While{(1)}
\If {$\overline{Z} \le Z$}
\State End the algorithm and output $V^{\star} = V$.
\Else
\State Compute $k(j') = \argmin_{k(j), j \in \{1, \ldots, J\}} r_{k(j)}$.
\State $V_{k(j'),j'} \gets 0$, $V_{k(j')-1,j'} \gets 1$,
\State $\overline{Z} \gets \overline{Z} - z_{k(j')} + z_{k(j')-1}$, $k(j') \gets k(j')-1$.
\EndIf
\EndWhile
\Else
\State End the algorithm and display `No feasible solution'.
\EndIf
\State \textbf{Output:} $V^{\star}$ of size $G \times J$, and the profit: $\mathrm{P} = \sum_{i=1}^G \left(c_i\sum_{j=1}^J  V^{\star}_{i,j} - d(z_i) \right)$.
\end{algorithmic}
\end{algorithm}


\subsection{Computational Complexity of the Greedy Algorithm} \label{SubSec:comp_complexity}
Now, we find the computational complexity of the proposed greedy algorithm, if the solution exists. First, we compute the values of $r_i, \forall i \in \{1, \ldots, G\},$ using a loop described in lines 4–7 of Algorithm~\ref{alg:greedy}. This computation results in the time complexity of $\mathcal{O}(G)$. Similarly, we find that the computational complexity of the loop described in lines 8–13 is $\mathcal{O}(J)$. Next, in the loop described in lines 15–24, we compute the argument minimizer in line 19 whose computational complexity is $\mathcal{O}(J)$, and this loop, in the worst case, executes till all $k(j), j\in \{1, \ldots, J\},$ become 1. This happens after $\mathcal{O}(G)$ times execution of the loop. Thus the computational complexity of the loop described in lines 15–24 is $\mathcal{O}(GJ)$. Hence, the total computational complexity of the proposed greedy algorithm in Algorithm~\ref{alg:greedy} is $\mathcal{O}(G+J+GJ)$. 

The computational complexity of finding the solution for the DAP using the brute-force search method is $\mathcal{O}(2^{GJ})$, since it uses all the possible binary matrices of size $G \times J$, sequentially, to compute the solution.

\begin{remark}
   The proposed greedy algorithm is highly efficient in terms of the computational complexity w.r.t. the brute-force search method, i.e.,  $\mathcal{O}(G+J+GJ) ~<<~\mathcal{O}(2^{GJ})$.
\end{remark}

The comparison study of the numerical solutions of the DAP using the proposed greedy algorithm and Gurobi software~\cite{gurobi} as a solver is shown in Section~\ref{SubSec:Simu_DAP}. 

\section{Simulation Results}\label{Sec:Simulations}
\begin{table}
\caption{Simulation hyper-parameters}
    \centering
    \begin{tabular}{|l|l|}
    \hline
        Number of matches used in training & 1580 \\ \hline
        Number of matches used in evaluation & 340 \\ \hline
        Number of epochs during training & 10 \\ \hline
        SNR & 6 dB \\
        \hline
        $\lambda$ & 0.3 \\
        \hline
        Learning rate & 0.001 \\
        \hline
        Dropout rate & 0.1 \\
        \hline
        Batch Size & 64 \\
        \hline
        Channel & AWGN \\ \hline
        Standard Deviation $\sigma$ in AWGN& 0.02\\ \hline
    \end{tabular}
    \label{tab:sim_param}
\end{table}

\begin{table}
\caption{Simulation settings for SemCom Encoder/Decoder layers}
    \centering
    \begin{tabular}{|c|c|c|c|}
    \hline
         & Layer & Units & Activation \\ \hline
         \multirow{3}{4em}{Transmitter} & Transformer Encoder (3) 
         & 128 (8 heads) & Linear \\ 
         & Dense & 256 & Relu \\ 
         & Reshape & 16 & Relu \\ \hline
         Channel & AWGN & None & None \\ \hline
         \multirow{3}{4em}{Receiver} 
         & Reshape & 256 & Relu \\ 
         & Dense & 128 & Relu \\ 
         & Transformer Decoder (3) 
         & 128 (8 heads) & Linear \\ \hline
    \end{tabular}
    \label{tab:SemCom_param}
\end{table}
In this section, we first provide the simulation results related to the designed SemCom system in Section~\ref{SubSec:Simu_SemCom} and then provide the  results related to the DAP in Section~\ref{SubSec:Simu_DAP}.

\subsection{The performance of SemCom System Model} \label{SubSec:Simu_SemCom}
First, we evaluate the performance of the text data transmission in terms of accuracy using BLEU score~\cite{papineni2002bleu}.\footnote{We defined the BLEU score in~\eqref{eq:BLEUScore}.} In our work, we use the dataset provided in~\cite{zhang2021soccer}. We parse the football commentary data of 1920 matches from the website \url{goal.com}. The considered football matches are from Union of European Football Associations (UEFA) Champions League, UEFA Europa League, and Premier League between 2016 and 2020. The simulations are performed in a computer with NVIDIA GeForce RTX 3090 GPU and Intel Core i9-10980XE CPU with 256GB RAM.

{\color{black} The simulation hyper-parameters used for plots in this section are shown in Table~\ref{tab:sim_param}. The simulation settings of the proposed SemCom system architecture (see Fig.~\ref{fig:SemComArchitecture}) consist of three transformer encoder and decoder layers with eight heads each. The dense layers in the transmitter and receiver are 256 units and 128 units, respectively. Similarly, the reshape layers in the transmitter and receiver are 16 units and 256 units, respectively. The linear activation functions are used in the encoders and decoders, whereas Relu activation functions are used in the rest of the layers. These settings are also listed in Table~\ref{tab:SemCom_param}.
} 

Let $\rho$ be the fraction of the total vocabulary $V$, which contains all the dataset words, to be added to $\mathcal{K}$. $\rho=0$ indicates that no additional vocabulary is added and the system is evaluated  only with the initial keyword set $\Omega_0$. Based on the way of adding the vocabulary words to $\Omega_0$, we propose two types of schemes. In the first type, $\rho |V|$ vocabulary words are uniformly chosen at random from $V$ and added to $\mathcal{K}$. In the second type, the words in $V$ are first arranged in the decreasing order of the frequency of appearances in the dataset, and then the first $\rho |V|$ vocabulary words are added to $\mathcal{K}$. We call these schemes as `RANDOM' and `ORDERED', respectively. 

The accuracy performances of both the schemes and a deep learning based SemCom system method named DeepSC~\cite{xie2021deep}, in terms of BLEU score vs. $\rho$, are shown in Fig.~\ref{fig:BLEU}. From the plot we can infer that even with $\rho=0$, the initial keyword set can produce a BLEU score of 0.55 (for 1-gram). This shows that the context-related keywords produce good results. Also, we see that as we add more vocabulary words to $\Omega_0$, the BLEU score increases. 
For the same value of $\rho$ and $n$, the ORDERED scheme performs better than the RANDOM scheme because of the addition of high frequency words.
And, in terms of different $n$-grams, BLEU score decreases as $n$ increases, which is an expected result. In comparison to the DeepSC scheme, the proposed schemes perform poorly in terms of accuracy but outperform it in terms of overhead reduction, as shown below.  

Next, we evaluate the performance of the proposed schemes, in terms of the transmission of average number of words per sentence, with respect to DeepSC~\cite{xie2021deep} and the results are shown in Fig.~\ref{fig:w_bar_vs_rho}. Let $\overline{W}$ denote the average number of words per sentence. From the plot we observe that both the schemes outperform DeepSC. Among the proposed schemes, for a given $\rho$ the RANDOM scheme outperforms the ORDERED scheme. This is because, in the ORDERED scheme high frequency words are added which increases the number of words to be encoded in the input data as compared with the RANDOM scheme. 

Now, we solve the optimization problem presented in~\eqref{eq:min_words} and~\eqref{eq:tau_thr}, with $\lambda=0$, using both the proposed schemes. For this purpose, we evaluate $\overline{W}$ vs. $\tau$ and the results are shown in Fig.~\ref{fig:w_bar_vs_tau}. From the plot we observe that both the schemes outperform DeepSC. Also, we see that the performance of both the schemes is same for a given accuracy threshold $\tau$. This is because, as shown in Fig.~\ref{fig:BLEU}, for a given value of $\rho \in (0,1)$, the ORDERED scheme outperforms the RANDOM scheme in terms of accuracy, whereas in Fig.~\ref{fig:w_bar_vs_rho}, the RANDOM scheme outperforms the ORDERED scheme in terms of overhead reduction. Hence, we can choose any one of the proposed methods to solve the optimization problem. 
\begin{figure}
\centering
\resizebox{0.88\columnwidth}{!}
{\includegraphics{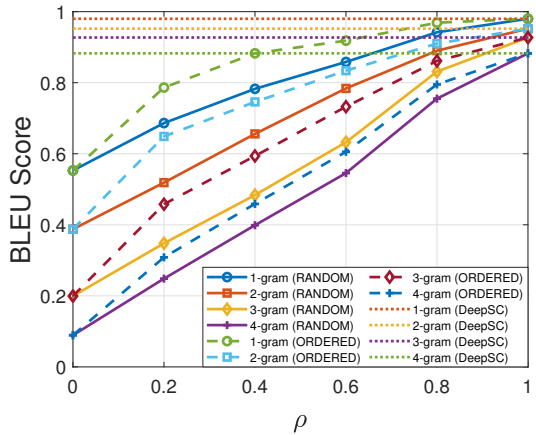}}
\caption{This plot shows the BLEU score vs. $\rho$ for different values of $n$-grams, where $n=\{1,2,3,4\}$, for the proposed schemes and the DeepSC scheme~\cite{xie2021deep}.}
\label{fig:BLEU}
\end{figure}

\begin{figure}
\centering
\begin{subfigure}{.24\textwidth}
\centering
\begin{adjustbox}{width = 1\columnwidth}
\includegraphics[width=0.99\textwidth]{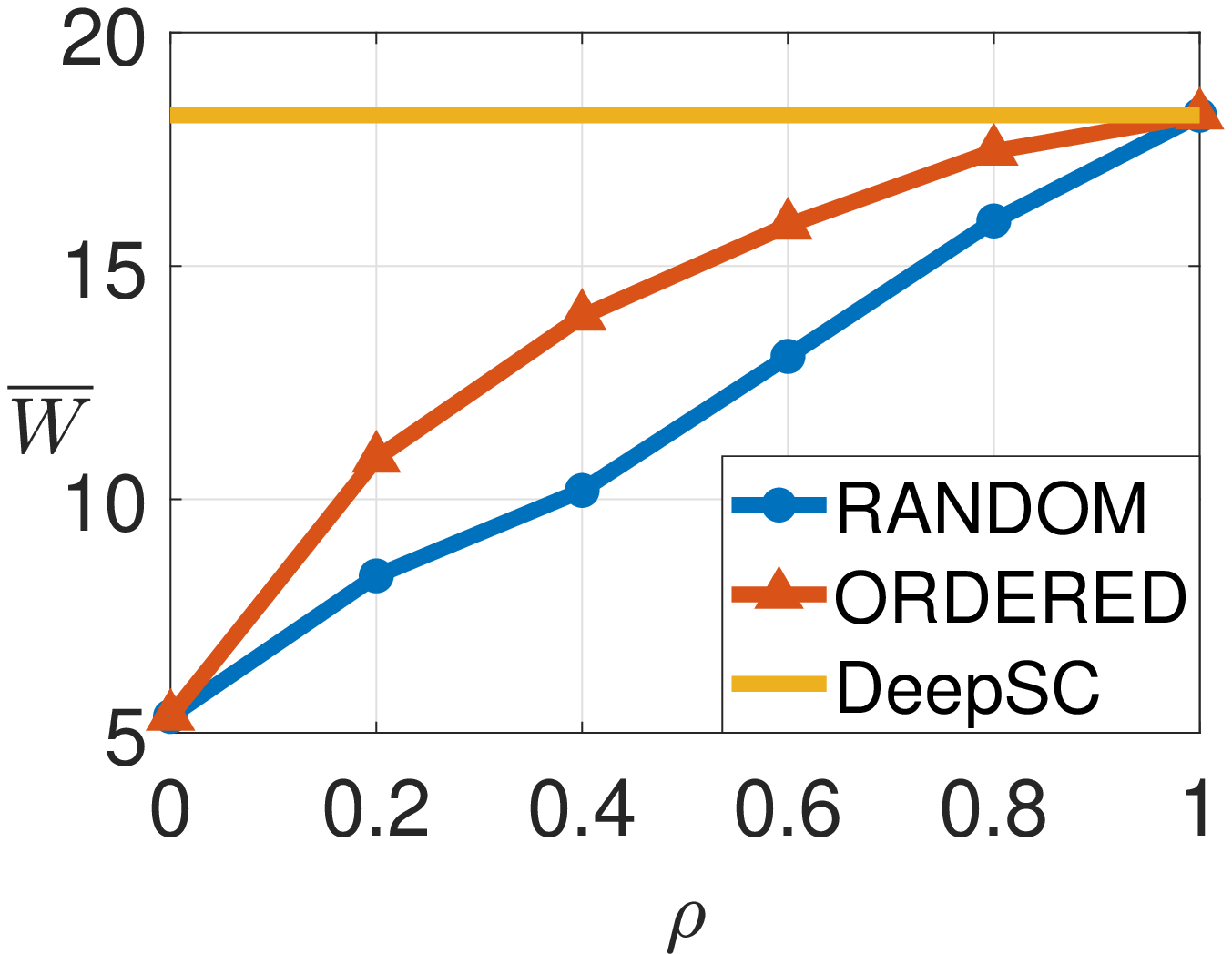}
\end{adjustbox}
\caption{}
\label{fig:w_bar_vs_rho}
\end{subfigure}%
\begin{subfigure}{.24\textwidth}
\centering
\begin{adjustbox}{width = 1\columnwidth}
\includegraphics[width=0.99\textwidth]{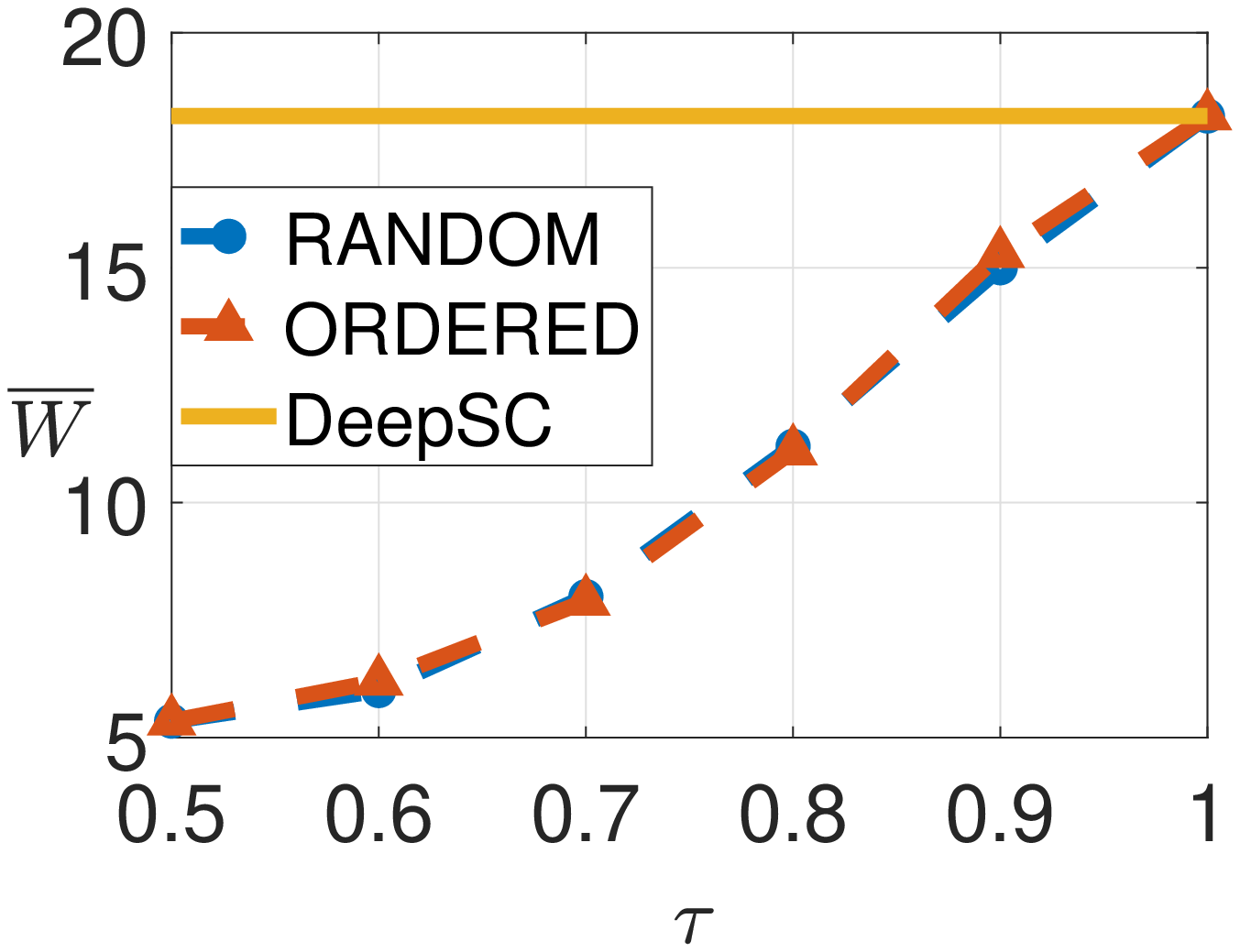}
\end{adjustbox}
\caption{}
\label{fig:w_bar_vs_tau}
\end{subfigure}
\caption{These plots show the average number of words per sentence vs. $\rho$ in the left plot and vs. $\tau$ in the right plot, respectively, for the proposed schemes and the DeepSC scheme~\cite{xie2021deep}.}
\label{fig:W_bar_plots}
\end{figure}

Now, we evaluate the performance of the proposed schemes, in terms of the semantic score (see~\eqref{eq:SemanticScore}), with respect to DeepSC~\cite{xie2021deep} and Joint Source-Channel Coding (JSCC) schemes~\cite{farsad2018deep}, and the results are shown in Fig.~\ref{fig:SemanticScore_vs_rho}. The trends are similar to that of Fig.~\ref{fig:BLEU}, and the proposed schemes outperform JSCC when $\rho$ approaches 0.8. Next, we compare the performance of the proposed scheme (with $\rho = 0.8$) in terms of different SNR values and the results are shown in Fig.~\ref{fig:SemanticScore_vs_SNR}. As expected, the semantic score increases as SNR increases due to a reduction in the noise effect for all the schemes but saturates for higher SNR values. When compared with DeepSC, the performance is slightly poor, but the proposed scheme follows the trends of DeepSC. But in comparison to JSCC, the proposed scheme outperforms it for all SNR values.
\begin{figure}
\centering
\begin{subfigure}{.24\textwidth}
\centering
\begin{adjustbox}{width = 1\columnwidth}
\includegraphics[width=0.99\textwidth]{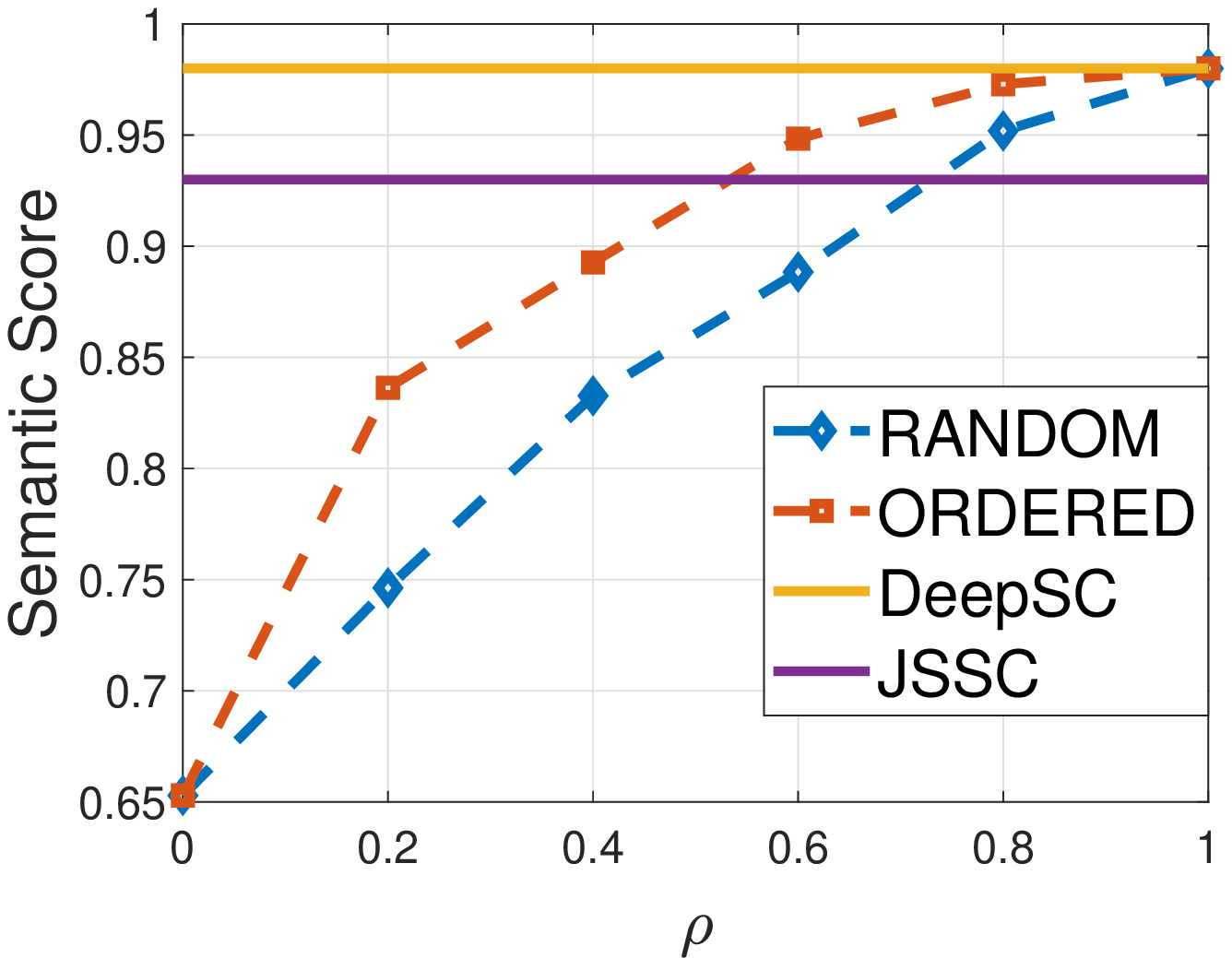}
\end{adjustbox}
\caption{}
\label{fig:SemanticScore_vs_rho}
\end{subfigure}%
\begin{subfigure}{.24\textwidth}
\centering
\begin{adjustbox}{width = 1\columnwidth}
\includegraphics[width=0.99\textwidth]{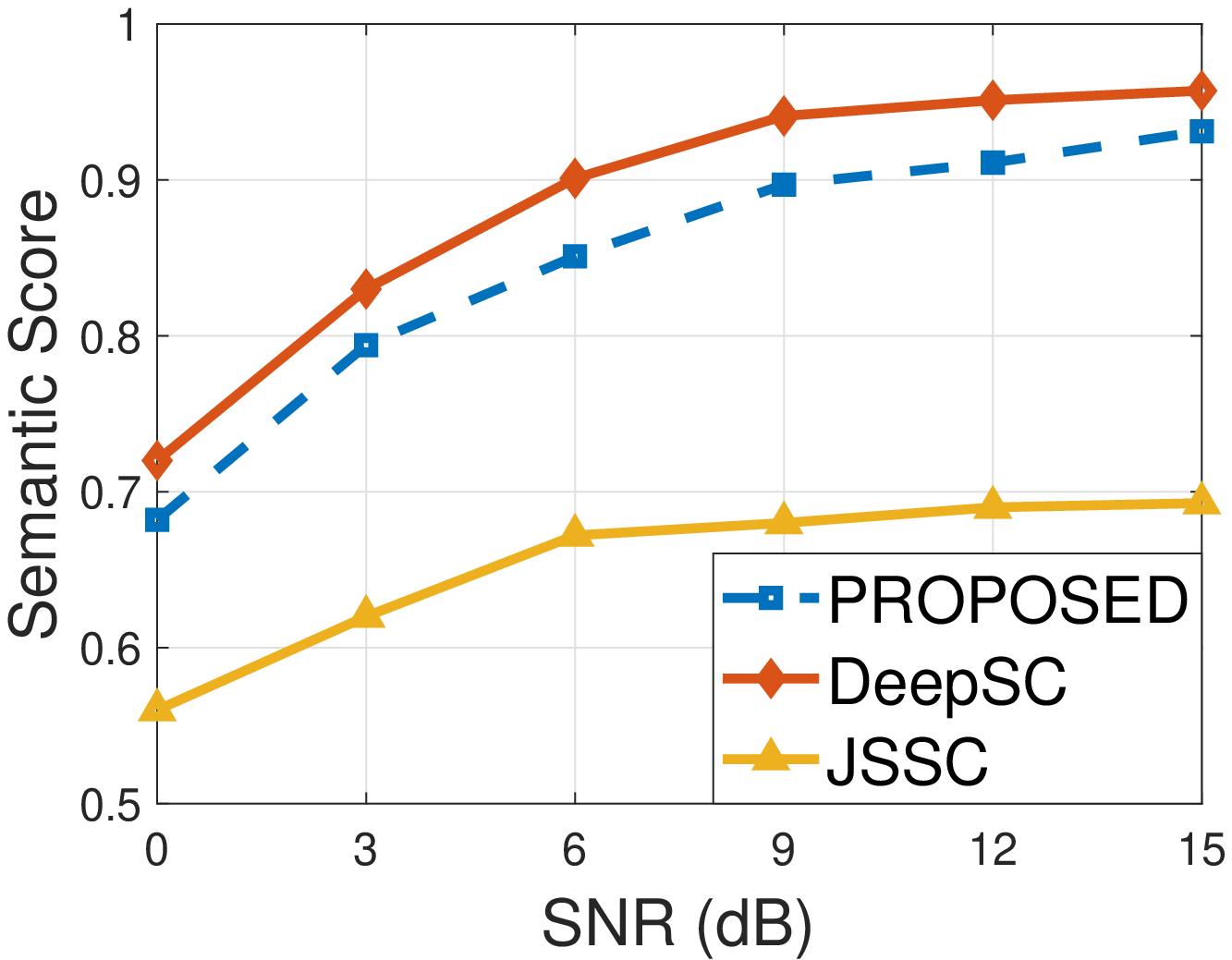}
\end{adjustbox}
\caption{}
\label{fig:SemanticScore_vs_SNR}
\end{subfigure}
\caption{These plots show the semantic score (SS) vs. $\rho$ in the left plot and vs. SNR (in dB) in the right plot, respectively, for the proposed schemes, DeepSC scheme~\cite{xie2021deep}, and JSCC~\cite{farsad2018deep}.}
\label{fig:SemanticScore_plots}
\end{figure}

Now, we compare the performance of our scheme with those of the DeepSC~\cite{xie2021deep} and adaptive~\cite{sana2022learning} schemes. Recall from Section~\ref{Sec:Performance_Analysis} that in the DeepSC and  proposed schemes, an average $n_0$ number of symbols are used for every word during encoding, and the adaptive scheme proposed in~\cite{sana2022learning} uses an adaptive method for choosing the number of symbols for every word that depends on the size of that word (see~\eqref{eq:q_il}). 
We show the comparisons among the proposed method and the schemes proposed in~\cite{xie2021deep,sana2022learning} in terms of $\widehat{\Psi}, \Psi_0,$ and $\Psi_\tau$, in Fig.~\ref{fig:sym_vs_rho}. From this plot, we observe that the proposed scheme transmits a significantly smaller number of symbols compared with both schemes, from $\rho=0$ to $\rho = 0.6$.
Next, we compare the performance of our scheme in terms of $\alpha_\tau^d$, $\alpha_\tau^a$, and $\alpha_\tau^p$ vs. the accuracy parameter $\tau$. The comparisons are shown in Fig.~\ref{fig:alpha_vs_tau}. From this plot, we observe that the proposed scheme outperforms both  schemes for accuracy levels {\color{black}up to} $82\%$. 
\begin{figure}
\centering
\begin{subfigure}{.24\textwidth}
\centering
\begin{adjustbox}{width = 1\columnwidth}
\includegraphics[width=0.99\textwidth]{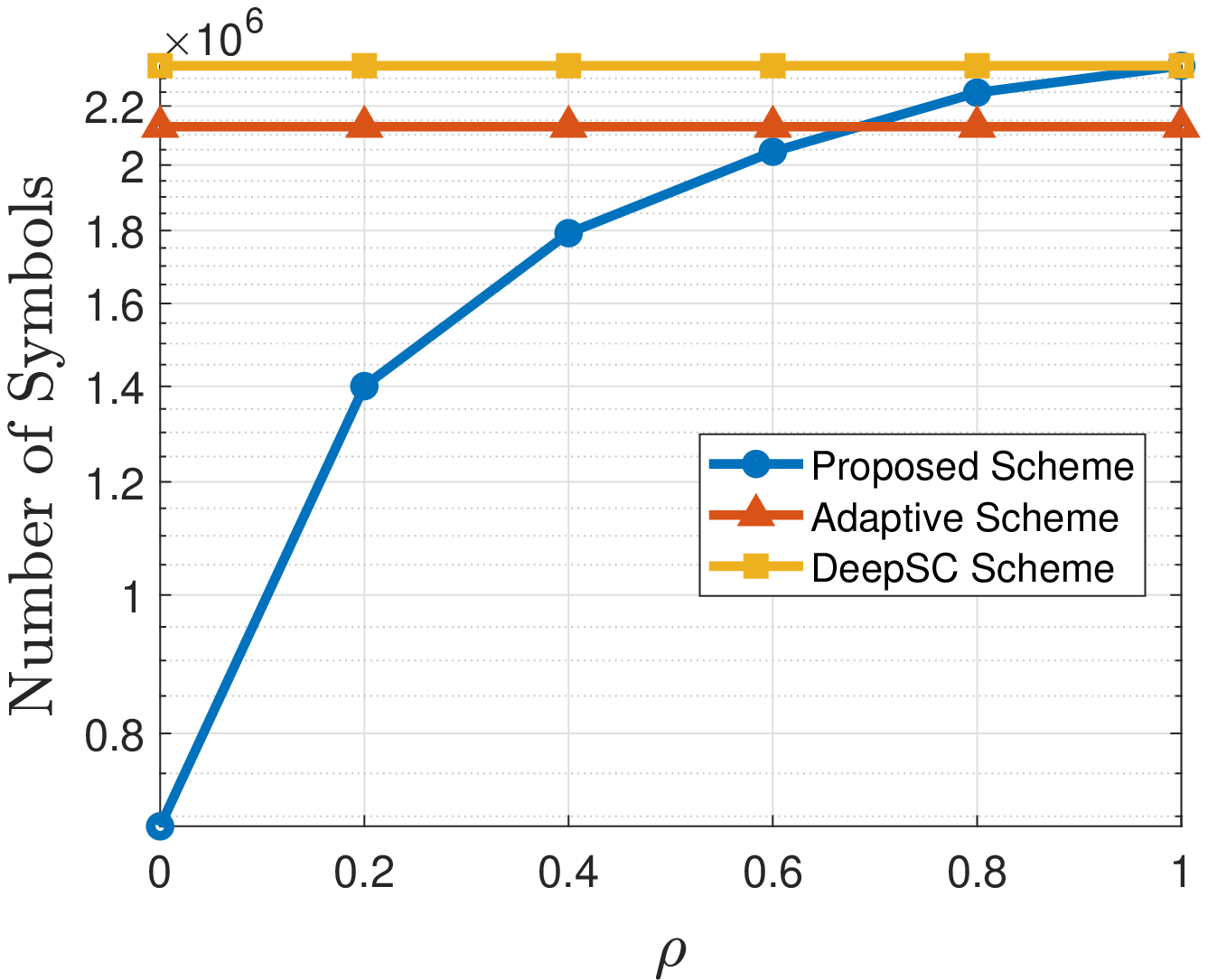}
\end{adjustbox}
\caption{}
\label{fig:sym_vs_rho}
\end{subfigure}%
\begin{subfigure}{.24\textwidth}
\centering
\begin{adjustbox}{width = 1\columnwidth}
\includegraphics[width=0.99\textwidth]{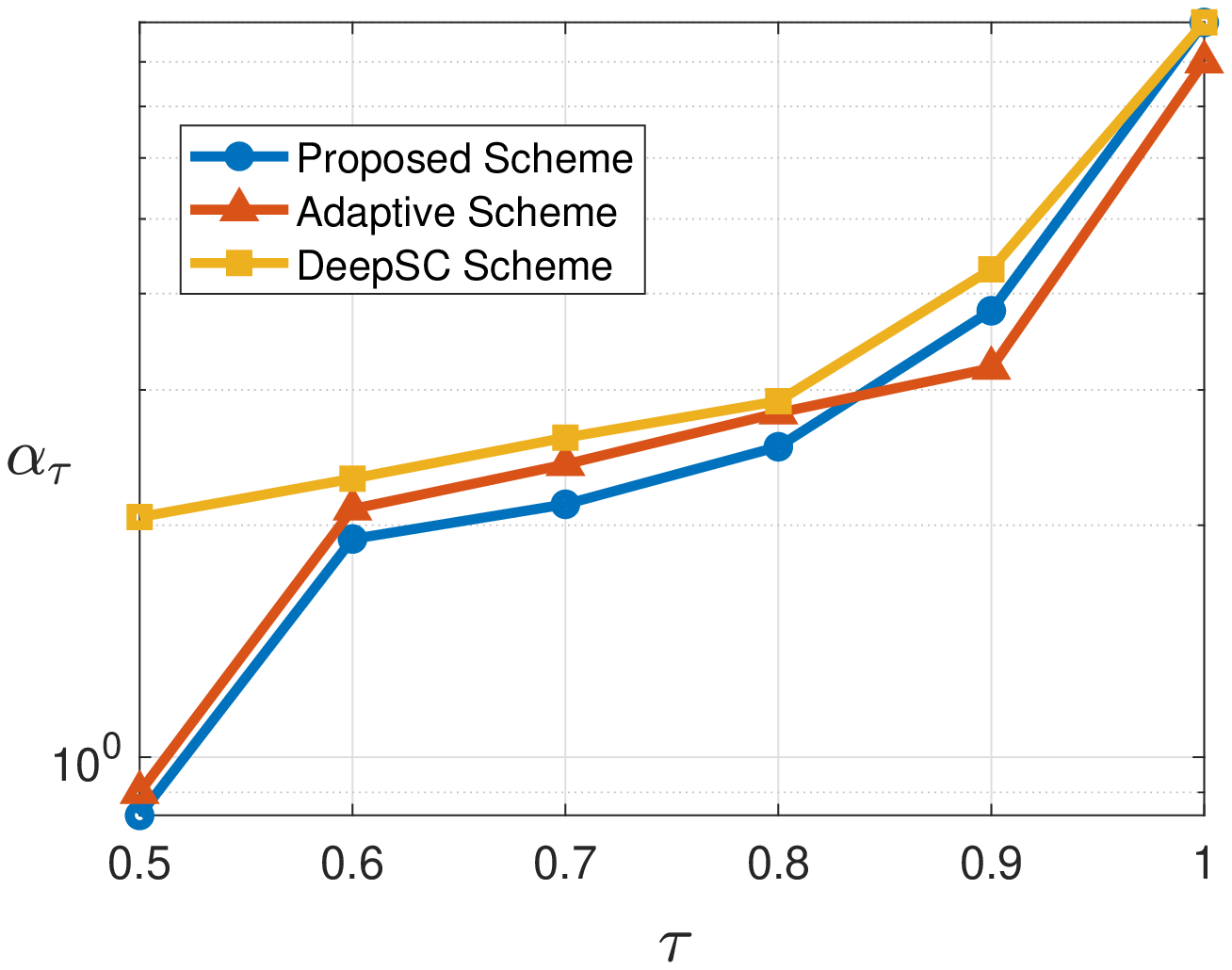}
\end{adjustbox}
\caption{}
\label{fig:alpha_vs_tau}
\end{subfigure}
\caption{The plot in left  shows the total number of symbols used in each of the schemes with respect to $\rho$.  We use $n_{\min} =1$, $n_0 =4$. The plot in right shows the $\alpha_\tau$ values in each of the schemes with respect to the given accuracy $\tau$.}
\label{fig:Avg_Sym_plots}
\end{figure}

\subsection{Solutions of the Data Allocation Problem} \label{SubSec:Simu_DAP}
Now, we present the simulation results related to the solution of the DAP defined in~\eqref{eq:max_profit}--\eqref{eq:integer_constraint}. We solve the DAP using three methods: \textit{Optimal}, \textit{Greedy}, and \textit{Greedy-cost}. We refer to the solutions obtained by the Gurobi software~\cite{gurobi} and the greedy algorithm (see Algorithm~\ref{alg:greedy}) as Optimal and Greedy, respectively. Similarly, the solution obtained by the algorithm, which is the same as that of the greedy algorithm, except that the argument minimizer minimizes the cost $c_i$ instead of the ratio $c_i/z_i$, $\forall i \in \{1, \ldots, G\}$, in line 19 of the proposed Algorithm~\ref{alg:greedy} is called greedy-cost. The primary goal of using the greedy-cost algorithm is to demonstrate numerically that maximizing profit by greedily changing only the costs does not produce better results than the proposed greedy algorithm. In our simulations, we have assumed that the data center has purchased a standard persistent disk (PD) from Google cloud~\cite{googlecloud} and has a memory capacity of $Z = 64$TB, minimum and maximum data sizes are 10GB and 100GB, respectively, and $G=20$. The costs and sizes are chosen uniformly at RANDOM from $[0,1]$ and $[10, 100]$, respectively, and the number of iterations in our simulations is 25. 
\begin{figure}
\centering
\begin{subfigure}{.24\textwidth}
\centering
\begin{adjustbox}{width = 1\columnwidth}
\includegraphics[width=0.99\textwidth]{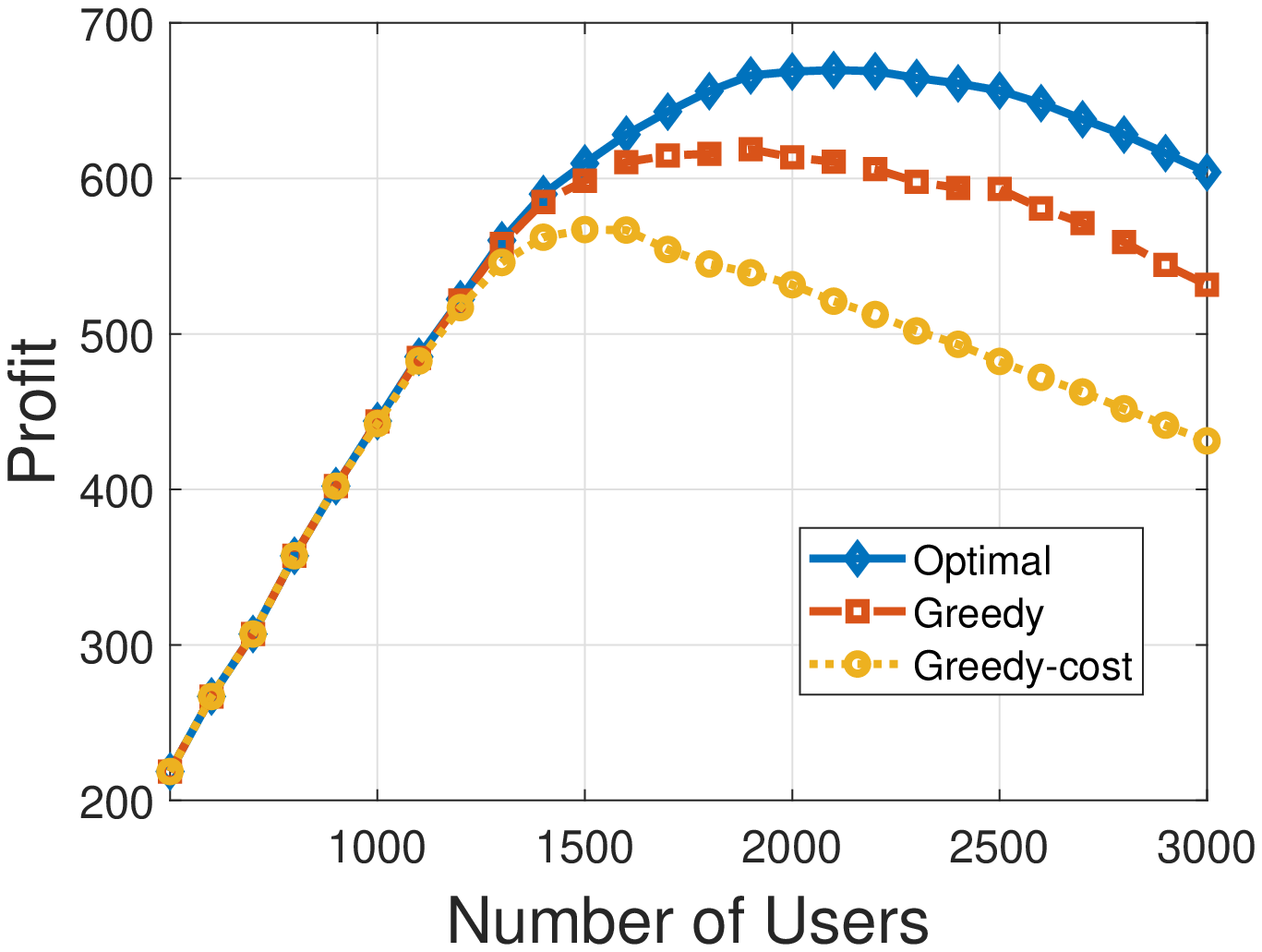}
\end{adjustbox}
\caption{}
\label{fig:Profit_vs_J}
\end{subfigure}%
\begin{subfigure}{.24\textwidth}
\centering
\begin{adjustbox}{width = 1\columnwidth}
\includegraphics[width=0.99\textwidth]{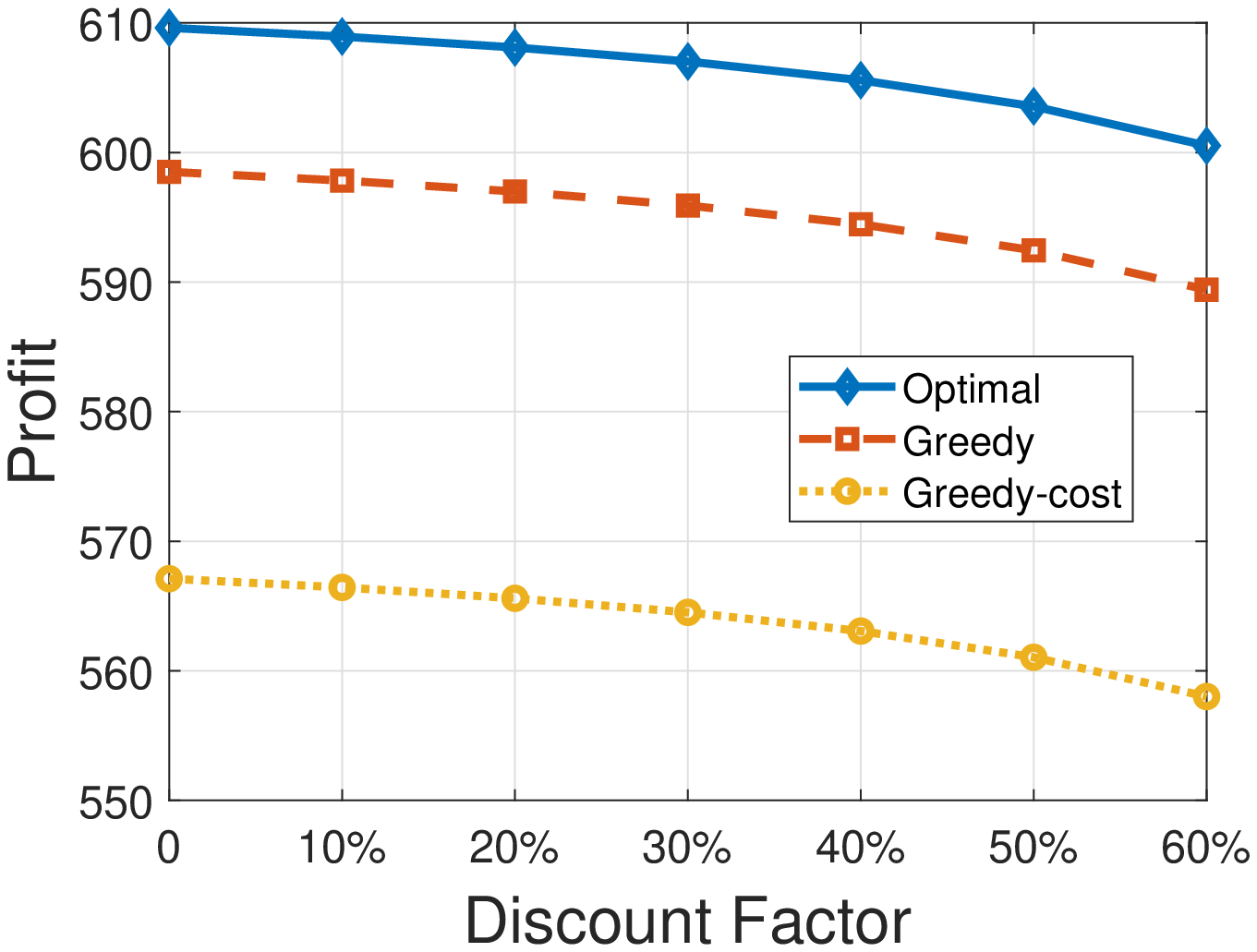}
\end{adjustbox}
\caption{}
\label{fig:Profit_vs_f}
\end{subfigure}
\caption{The plot in left shows the total profit gained using all three algorithms with respect to the number of subscribed users $J$. The maximum profit is observed for $J=2100$, and the profit computed by the proposed greedy algorithm (respectively, greedy-cost algorithm) at the same value of $J$ is $90.54\%$ (respectively, $77.76\%$) of the optimal maximum profit. The plot in right shows the total profit gained using all three algorithms with respect to the discount factor. The average fall of the profit with each discount factor for the optimal, greedy, and greedy-cost algorithms is $0.653\%$, $0.666\%$, and $0.706\%$, respectively. This shows that the discount factor does not affect the profit significantly.  Hence, there is a room to attract more subscribers without loosing the significant profit. We use $J=1500$ in this case.}
\label{fig:profit_plots}
\end{figure}


First, we compute the total profit gained by the data center with respect to the number of users it serves, and the results obtained by all three methods are shown in Fig.~\ref{fig:Profit_vs_J}. From the plot, we can observe that for a set of a smaller number of users, in particular from $J=500$ to $J=1200$ in our case, the profit computed by all three methods is the same. This is because every method is successful in allocating the best possible category data to every user without violating the size constraint~\eqref{eq:size_constraint}, due to the small number of users. As the number of users increases the profit obtained by \textit{Optimal} starts outperforming both greedy algorithms. The proposed greedy algorithm solution closely follows the optimal solution, but the greedy-cost algorithm solution starts moving away significantly from the optimal solution. This is because the greedy-cost algorithm only accounts for the cost maximization without bothering about the data sizes, which results in violation of the size constraint~\eqref{eq:size_constraint} more often than the proposed greedy algorithm, which accounts for both the costs and the data sizes. The plot in Fig.~\ref{fig:Profit_vs_J} also shows that the profit increases initially for all three algorithms, then reaches its maximum value and begins to decrease again.  

Next, we evaluate the performance of the algorithms in terms of profit gained with respect to the discount factor, which is the percentage discount given by the data center to its users in comparison to the purchase price of the data to attract more new subscribers. The results are shown in Fig.~\ref{fig:Profit_vs_f}. As expected, the optimal solution outperforms both greedy algorithms, and also, the proposed greedy algorithm outperforms the greedy-cost algorithm, for all discount factors. Also, as the amount of discount increases, the profit for a fixed number of users reduces, which is also along expected lines. 

\begin{figure}
\centering
\resizebox{0.9\columnwidth}{!}
{\includegraphics{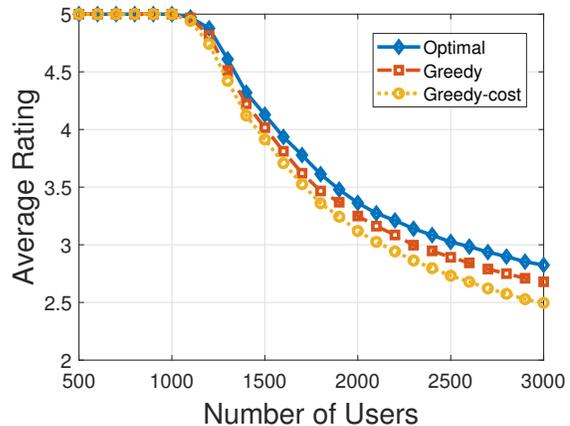}}
\caption{This plot shows the average rating given by users for the data center's service using each of the algorithms w.r.t. the number of users. For $J=2100$ users, where profit is maximized, the average ratings are 3.27, 3.18, and 3.02 for services provided using the optimal, greedy, and greedy-cost algorithm solutions, respectively. This result demonstrates that the average rating provided for the optimal solution is not significantly higher compared to that of the proposed greedy algorithm solution. 
}
\label{fig:AvgRating_vs_J}
\end{figure}
Now, we evaluate the users' satisfaction with the service provided by the data center by using their ratings. The ratings are provided by users using one of the numbers between 1 and 5, where 1 and 5 signify the worst and best user experiences, respectively. Let $\bar{i}(j),\Tilde{i}(j) \in \{1, \ldots, G\},$ be the quantities such that $U_{\bar{i}(j),j} = 1$ and $U_{\bar{i}(j)+1,j} = 0$, and $V^{\star}_{\Tilde{i}(j),j}=1,~j \in \{1, \ldots, J\}$. Let $\mathrm{SL}(j)$ denote the \textit{satisfaction level} of the users $j \in \{1, \ldots, J\}$. We define satisfaction level as $\mathrm{SL}(j) = \bar{i}(j) - \Tilde{i}(j),~j \in \{1, \ldots, J\}$. For the purpose of evaluation, we assume that the ratings and satisfaction levels are related as shown in Table~\ref{tab:sat_level_rating}.\footnote{{\color{black} The relation between satisfaction level of user $j$ and user $j^{th}$ ratings is studied in different contexts, such as video streaming~\cite{karim2019quality}, multimedia communications~\cite{ghinea1999approach}, information retrieval~\cite{al2010review}, audio transmission~\cite{wilson2000investigating}, speech transmission~\cite{watson2000good}, m-Commerce~\cite{ghinea2004user}, etc. The satisfaction level measures quality of service (QoS) provided by the data center, whereas user ratings measure quality of experience (QoE) experienced by the users. The values of satisfaction levels, scaled to 0–20, and their corresponding user ratings, scaled to 1–5, shown in Table III, are compatible with those of the work presented in~\cite{kawalek1995user}. A user gives a rating of 5 for a service in which the data center provides the same quality content that the user expects. Then the user tends to provide poor ratings as the data center provides poor-quality products.}} The plot in Fig.~\ref{fig:AvgRating_vs_J} shows the average rating provided by the subscribed users for the services provided by all three algorithms. From the plot we observe that all users provide the rating of 5 to the service when number of subscribers is low, i.e., $500 \le J \le 1100$ in our example. This is due to the lower number of subscribers, which resulted in the best possible category of data allocation based on each subscriber's budget. However, as $J$ increases beyond $1100$, every algorithm starts allocating lower level categories of data to some of the subscribers so that the size constraint~\eqref{eq:size_constraint} is not violated.  Hence, there is a fall in the average rating.  

The histograms of the ratings provided by the users are shown in Fig.~\ref{fig:J_vs_SL}. These plots support the observation that when $J$ is small, the size constraint~\eqref{eq:size_constraint} is easily satisfied, and thus the best possible category of data is allocated to each user. Conversely, when $J$ is large, to satisfy the size constraint~\eqref{eq:size_constraint} algorithms tend to allocate a lower category of data to users, resulting in lower ratings. 
\begin{table}
\caption{Relation between the satisfaction levels and user ratings}
    \centering
    \begin{tabular}{|c|c|c|c|c|c|}
    \hline
        Satisfaction Level & 0 & 1-2 & 3-5 & 6-10 & 11-20 \\ \hline
        User Rating & 5 & 4 & 3 & 2& 1 \\ \hline
    \end{tabular}
    \label{tab:sat_level_rating}
\end{table}

\begin{figure}
\centering
\begin{subfigure}{.24\textwidth}
\centering
\begin{adjustbox}{width = 1\columnwidth}
\includegraphics[width=0.99\textwidth]{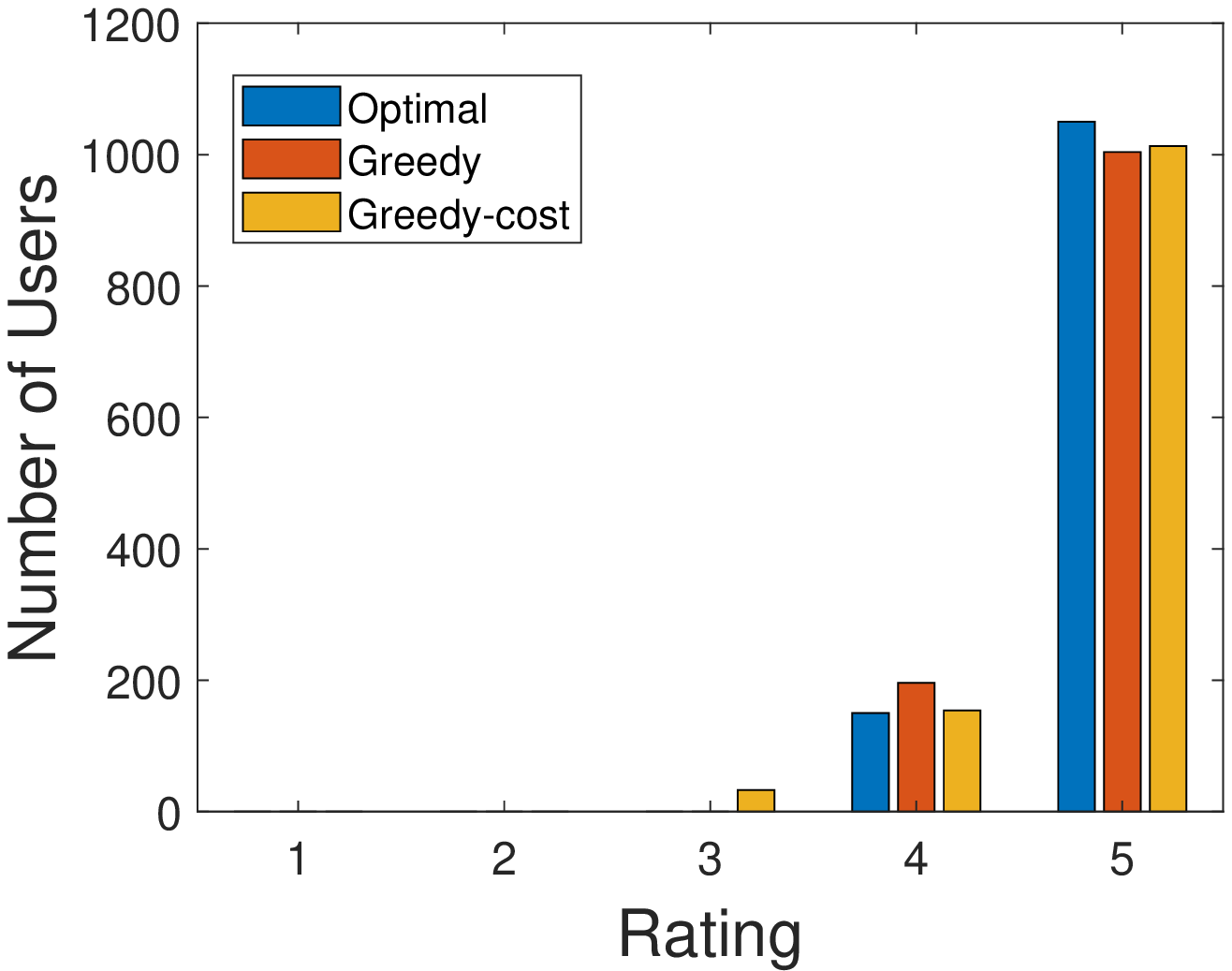}
\end{adjustbox}
\end{subfigure}%
\begin{subfigure}{.24\textwidth}
\centering
\begin{adjustbox}{width = 1\columnwidth}
\includegraphics[width=0.99\textwidth]{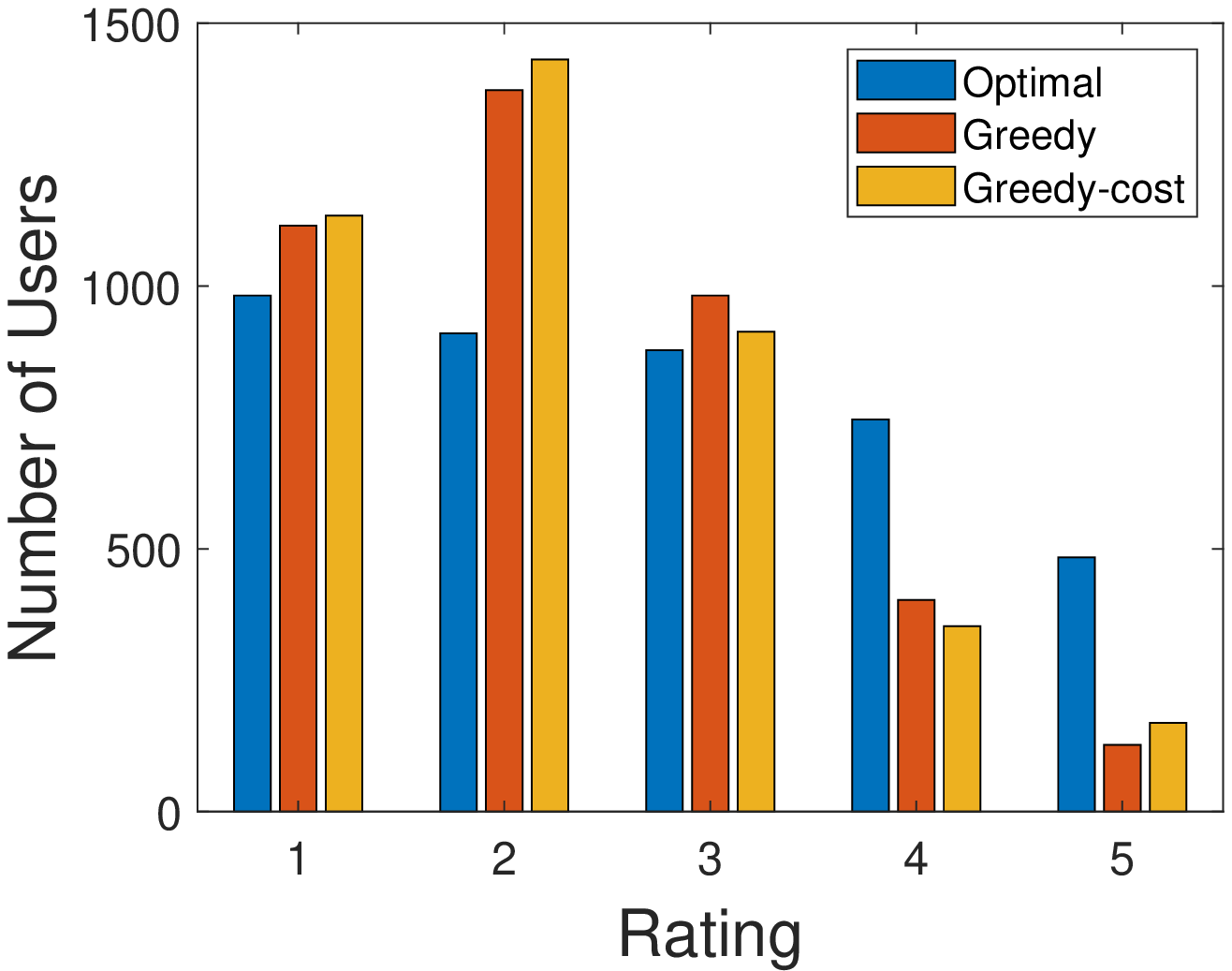}
\end{adjustbox}
\end{subfigure}
\caption{These plots show the histogram of the ratings provided by the users for each of the algorithms. The left figure shows the results for a small number of users, $J=1200$, whereas the right figure shows the results for a large, $J=4000$, number of users.}
\label{fig:J_vs_SL}
\end{figure}

\section{Conclusions and Future Work}\label{Sec:Conclusions}
In this paper, we first extracted relevant keywords from the dataset using the shared knowledge base. Then, using the received keywords and the shared knowledge, we designed an auto-encoder and auto-decoder that only transmit these keywords and, respectively, recover the data. We provided analytical comparisons of the proposed scheme versus the DeepSC~\cite{xie2021deep} and the adaptive~\cite{sana2022learning} schemes in terms of accuracy vs. overhead reduction trade-off and cost comparisons. We computed the accuracy of the reconstructed sentences at the receiver quantitatively. We demonstrated through simulations that the proposed methods outperform a state-of-the-art method in terms of the average number of words per sentence. The designed SemCom system is then applied to a realistic scenario in which a cloud server and a data center serve as transmitter and receiver, respectively. We formulated a data allocation problem (DAP), in which the data center optimally allocates various categories of datasets received from the cloud server to its subscribers. We proved that the DAP belongs to a class of NP-complete problems and proposed a greedy algorithm for solving it. Furthermore, we have numerically demonstrated that the solutions of the proposed greedy algorithm, in terms of profits, are $90\%$ of the optimal solutions. 

{\color{black}In this paper, we focused solely on the text dataset; however, similar SemCom system design approaches can be proposed in the future for other types of datasets such as images, audio, and video. Also, a real-time DAP can be formulated by considering the dynamic storage facility using cache memories at the data center in place of a static storage facility. Another direction for future research is to address the problem of the dynamic arrival and departure of subscribers for a service at the data center and analyze how the data center handles it. Finally, another open problem is to design an approximation algorithm with a provable approximation ratio for the DAP.}

\appendices

\section{Proof of Theorem~\ref{Thm_DAP}} \label{Apdx_Thm_DAP}
The decision version of the DAP is as follows: ``Given a number $L$, does there exist a binary matrix $V$, of size $G \times J$, that satisfy the constraints~\eqref{eq:size_constraint}-\eqref{eq:user_constraint} such that $ \sum_{i=1}^G \!\left(\!c_i\sum_{j=1}^J \! V_{i,j} \!- \!d(z_i) \right)$ $ \ge L$''? Given $V$, we can check in polynomial time whether it satisfies~\eqref{eq:size_constraint}-\eqref{eq:user_constraint} and whether
$\sum_{i=1}^G \left(c_i\sum_{j=1}^J  V_{i,j} - d(z_i) \right) \ge L$. Thus the DAP is in class NP~\cite{kleinberg2006algorithm}. By using~\eqref{eq:mi_Vij}, we simplify the expressions~\eqref{eq:max_profit} and~\eqref{eq:size_constraint} as following:
\begin{subequations}
    \begin{align}
        \max_{m_i, i \in \{1,\ldots, G\}} 
        & \sum_{i=1}^G c_i m_i - d(z_i) \label{eq:cost_max}\\
        & \sum_{i=1}^G z_i m_i \le Z. \label{eq:weight_constraint2}
    \end{align}
\end{subequations}
We now show that the DAP is NP-complete by reducing the knapsack problem (KP), which has been shown to be NP-complete~\cite{kleinberg2006algorithm}, to it. 

Let us consider the following KP: We want to pack $n$ different types of items in a knapsack which can withstand a maximum weight of $W$. Each item of type $i \in \{1,\ldots, n\}$ is categorised with two parameters: a weight $w_i$ and a value $v_i$. The goal of the KP is to find a set of items that produce the maximum possible value, with the restriction that the total weight of the set should not exceed $W$. It is also written as follows:
\begin{subequations}
    \begin{align}
        \max_{x_i, i \in \{1,\ldots, n\}} 
        & \sum_{i=1}^n v_i x_i \label{eq:value_max}\\
        & \sum_{i=1}^n w_i x_i \le W, \label{eq:weight_constraint}\\
        & x_i \in \{0, 1, \ldots, \}, 
    \end{align}
\end{subequations}
where $x_i, i \in \{1,\ldots, n\},$ denote the number of items of type $i$. The decision version of the KP is as follows: ``Given a number $L$, is it possible to achieve $\sum_{i=1}^n v_i x_i \ge L$  without exceeding the total weight constraint $W$''? Let us denote the decision version inequalities of both the problems as following:
\begin{subequations}
    \begin{align}
       \mathrm{D}_1 :& ~\sum_{i=1}^G c_i m_i - d(z_i) \ge L, \\
       \mathrm{D}_2 :& ~\sum_{i=1}^n v_i x_i \ge L.
    \end{align}
\end{subequations}

Now, we show that the KP is polynomial-time reducible to DAP, i.e., KP $<_p$ DAP. Consider the instance of the KP stated in the previous paragraph. From this instance, we construct the following instance of the DAP. In this instance, for every $i \in \{1,\ldots, G\}$, the selling costs are equivalent to the values, i.e., $c_i = v_i$; data sizes are equivalent to the weights, i.e., $z_i = w_i$; number of users with data category $i$ is equivalent to the number of items of type $i$, i.e., $m_i = x_i$; and the constants $Z = W$ and $G = n$. In this instance, we assume that every user is allocated only one category of data that satisfies  the constraint~\eqref{eq:user_constraint}. Also, we assume that the budget of every user is higher than the maximum possible selling cost of all categories of data, i.e., $b_j \ge c_G, \forall j \in \{1,\ldots, J\}$. This assumption leads to $U_{i,j} = 1, \forall i \in \{1,\ldots, G\}, j \in \{1,\ldots, J\}$ (see~\eqref{eq:u_ij}). Now, we see that the constraint~\eqref{eq:category_constraint}, i.e., $\sum_{i=1}^G \sum_{j=1}^J V_{i,j} = \sum_{j=1}^J \sum_{i=1}^G V_{i,j} = \sum_{j=1}^J 1 = J$ is automatically satisfied due to~\eqref{eq:user_constraint}, i.e., $\sum_{i=1}^G V_{i,j} = 1, \forall j \in \{1,\ldots, J\}$.

Given this instance, we ask: does there exist a binary matrix $V$ that satisfy the constraints~\eqref{eq:size_constraint}-\eqref{eq:user_constraint} and  $\mathrm{D}_1$? We claim that
the answer is yes if and only if the answer to the question in the KP instance is yes. The necessity part is proved as follows. If the answer to the above question is yes then there exists a binary matrix $V$ that satisfy the constraints~\eqref{eq:size_constraint}-\eqref{eq:user_constraint}, \eqref{eq:mi_Vij}, and  $\mathrm{D}_1$. From the assumptions we made in the previous paragraph, i.e., for every $i \in \{1,\ldots, G\}$, $c_i = v_i$, $z_i = w_i$, $m_i = x_i$, $Z = W$ and $G = n$, we see that the constraint~\eqref{eq:weight_constraint} and $\mathrm{D}_2$ are satisfied since~\eqref{eq:weight_constraint2} and  $\mathrm{D}_1$ are satisfied, respectively. This proves necessity.

To prove sufficiency, suppose the answer to the question in the KP is yes. Now, let us assume that for every $i \in \{1,\ldots, n\}$, $v_i = c_i$, $w_i = z_i$, $x_i = m_i = \sum_{j=1}^J  V_{i,j}$, $W = Z$, and $n = G$. Now we see that the constraint~\eqref{eq:size_constraint} and $\mathrm{D}_1$ are satisfied since~\eqref{eq:weight_constraint} and  $\mathrm{D}_2$ are satisfied, respectively. The constraints~\eqref{eq:category_constraint}  and~\eqref{eq:user_constraint} are satisfied due to the construction of the instance in the DAP. This proves sufficiency and the result follows.

\bibliographystyle{ieeetr}
\bibliography{references}
\begin{IEEEbiography}
[{\includegraphics[width=1in,height=1.25in,clip,keepaspectratio]{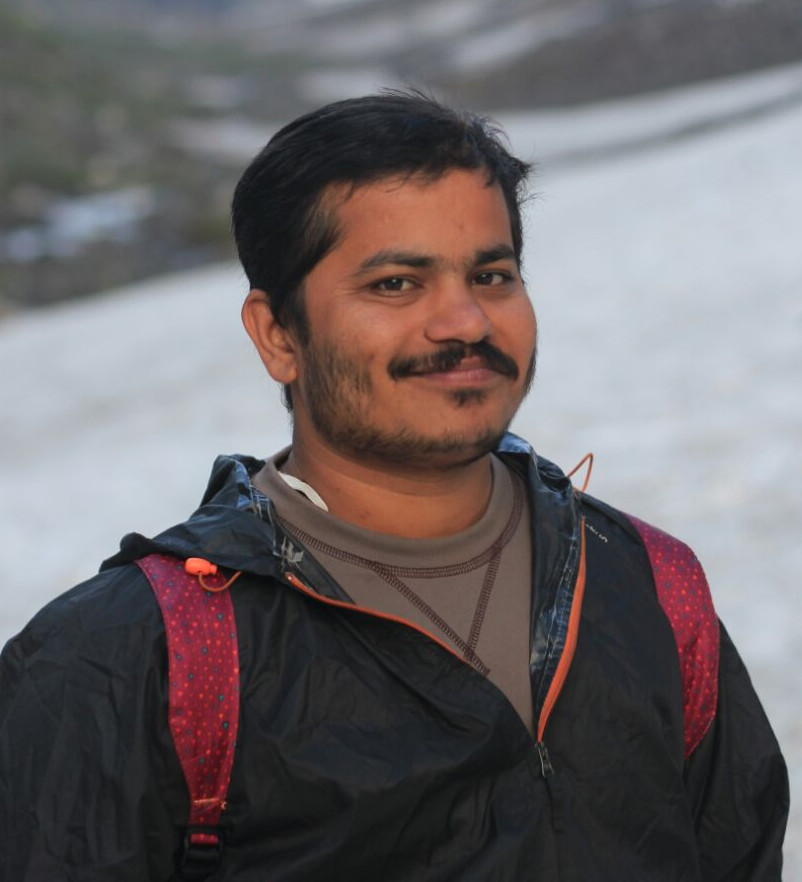}}]
 {Sachin Kadam} received the B.Eng. degree in electronics and communication engineering from the People’s Education Society Institute of Technology, Bengaluru, India, in 2007, the M.Tech. degree in electrical engineering from the Indian Institute of Technology (IIT) Kanpur, India, in 2012, and the Ph.D. degree from IIT Bombay, India, in 2020. He is currently a Postdoctoral Researcher with Prof. Dong In Kim with Sungkyunkwan University, Suwon-si, Gyeonggi-do, Republic of Korea. His research interests include the design and analysis of wireless and M2M networks, semantic communications, differential privacy, and learning. He was a recipient of Scholarship Foundation for Excellence, California, USA, during the B.Eng. degree.
\end{IEEEbiography}

\begin{IEEEbiography}
[{\includegraphics[width=1in,height=1.25in,clip,keepaspectratio]{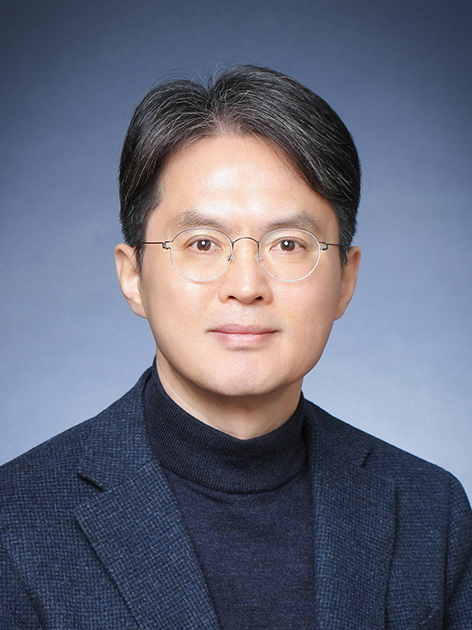}}]
 {Dong In Kim} (Fellow, IEEE) received the Ph.D. degree in electrical engineering from the
University of Southern California, Los Angeles, CA, USA, in 1990. He was a Tenured
Professor with the School of Engineering Science, Simon Fraser University, Burnaby, BC,
Canada. He is currently a Distinguished Professor with the College of Information and
Communication Engineering, Sungkyunkwan University, Suwon, South Korea. He is a
Fellow of the Korean Academy of Science and Technology and a Member of the National
Academy of Engineering of Korea. He was the first recipient of the NRF of Korea
Engineering Research Center in Wireless Communications for RF Energy Harvesting from
2014 to 2021. He received several research awards, including the 2023 IEEE ComSoc Best
Survey Paper Award and the 2022 IEEE Best Land Transportation Paper Award. He was
selected the 2019 recipient of the IEEE ComSoc Joseph LoCicero Award for Exemplary
Service to Publications. He was the General Chair of the IEEE ICC 2022, Seoul. Since 2001,
he has been serving as an Editor, an Editor at Large, and an Area Editor of Wireless
Communications I for IEEE Transactions on Communications. From 2002 to 2011, he served
as an Editor and a Founding Area Editor of Cross-Layer Design and Optimization for IEEE
Transactions on Wireless Communications. From 2008 to 2011, he served as the Co-Editor-
in-Chief for the IEEE/KICS Journal of Communications and Networks. He served as the
Founding Editor-in-Chief for the IEEE Wireless Communications Letters from 2012 to 2015.
He has been listed as a 2020/2022 Highly Cited Researcher by Clarivate Analytics.

 \end{IEEEbiography}
\end{document}